\documentclass[12 pt]{amsart}

\usepackage{amscd,amsfonts,amssymb,amsmath}
\usepackage{tikz-cd}
\usepackage{xcolor}
\usepackage{placeins}
\usepackage{hyperref}

\newtheorem{theorem}{Theorem}[section]
\theoremstyle{definition}

\newtheorem{corollary}[theorem]{Corollary}
\newtheorem{lemma}[theorem]{Lemma}
\newtheorem{proposition}[theorem]{Proposition}
\newtheorem{question}[theorem]{Question}
\newtheorem{problem}[theorem]{Problem}
\newtheorem{definition}[theorem]{Definition}
\newtheorem{example}[theorem]{Example}
\newtheorem{remark}[theorem]{Remark}
\numberwithin{equation}{subsection}

\newcommand{\Alex}{\operatorname{Alex}}
\newcommand{\Aut}{\operatorname{Aut}}

\newcommand{\Conj}{\operatorname{Conj}}

\newcommand{\Core}{\operatorname{Core}}

\newcommand{\Inn}{\operatorname{Inn}}

\newcommand{\T}{\operatorname{T}}

\newcommand{\R}{\operatorname{R}}

\newcommand{\Z}{\mathbb{Z}}

\newcommand{\id}{\mathrm{id}}

\newcommand{\SE}[1][n]{${#1}$-SE}

\newcommand{\veq}{\mathrel{\rotatebox{90}{$=$}}}

\newcommand{\ms}{\ast}
\newcommand{\mc}{\circ}
\newcommand{\mL}{\triangleleft}
\newcommand{\mR}{\triangleright}
\newcommand{\os}{\mathbin{\overline{\ast}}}
\newcommand{\oc}{\mathbin{\overline{\circ}}}
\newcommand{\oL}{\mathbin{\overline{\triangleleft}}}
\newcommand{\oR}{\mathbin{\overline{\triangleright}}}

\newcommand{\dotss}{ ...\, }

\renewcommand{\lg}{\langle}
\newcommand{\rg}{\rangle}
\renewcommand{\a}{\alpha}
\renewcommand{\b}{\beta}
\newcommand{\g}{\gamma}
\renewcommand{\d}{\delta}

\title{Set-theoretical solutions of simplex equations}
\author[V.~Bardakov, B.~Chuzinov, ...]{V.~Bardakov, B.~Chuzinov,  I.~Emel'yanenkov, M.~Ivanov, T.~Kozlovskaya, and V.~Leshkov}
\date{\today}

\begin{document}
\maketitle
\begin{abstract}
The $n$-simplex equation ($n$-SE) was introduced by A. B. Zamolodchikov as a generalization of the Yang--Baxter equation, which is the $2$-simplex equation in these terms. In the present paper  we suggest some general approaches to constructing solutions of $n$-simplex  equations, describe some types of solutions, introduce an operation which under some conditions allows us to construct a solution of  $(n+m+k)$-SE from  solution of  $(n+k)$-SE and  $(m+k)$-SE. 
We consider the tropicalization of rational solutions and discuss a way to generalize it. We prove that if a group $G$ is an extension of a group $H$ by a group $K$, then  we can find a solution of the $n$-SE on $G$ from solutions of this equation on $H$ and on $K$. Also, we find solutions of the parametric Yang-Baxter equation on $H$ with parameters in $K$. For studying solutions of the 3-simplex equations we introduce algebraic systems with ternary operations and give examples of these systems which gives solutions of the $3$-SE. 
We find all elementary verbal solutions of the  $3$-SE  on free groups.


\end{abstract}
\maketitle
\tableofcontents

\section{Introduction}

A solution of the quantum Yang--Baxter equation (YBE) is a linear map $R : V \otimes V \to V\otimes V$ satisfying 
\begin{equation}
\label{YB}
	R_{12} R_{13} R_{23} = R_{23} R_{13} R_{12},
\end{equation}
 where  $V$ is a vector space and $R_{ij} : V \otimes V \otimes V  \to  V \otimes V \otimes V$ acts as $R$ on the $i$-th and $j$-th tensor factors and as the identity on the remaining factor. 
V.~M.~Bukhshtaber \cite{Bukh} called the map $R$, which satisfies the YBE, a Yang-Baxter map. 

The  Yang-Baxter equation, or the 2-simplex equation, or the triangle equation, is one of the basic  equations in mathematical physics and in  low dimension topology.
It lies in the foundation of the theory of quantum groups, solvable models of statistical mechanics, knot theory and braid theory. 
At first, the YBE appears  in the paper of C.~N.~Yang  \cite{Yang} on the many-body problem. Later R.~J.~Baxter \cite{Bax} introduced  this equation for the study of solvable vertex models in statistical mechanics as the condition of commuting transfer matrices. Another derivation of the YBE follows from the factorization of the $S$-matrix in a $1+1$ dimensional Quantum Field Theory (see papers of  A.~B.~Zamolodchikov \cite{Z,Z-1}). The YBE is also essential in  Quantum Inverse Scattering Method for integrable systems \cite{STF, TF}.  

The Yang-Baxter equation is equivalent to a system of $n^6$ algebraic  cubic  equations for $n^4$ variables, where $n = \dim V$. So even in the two-dimensional case, one has 64 equations for 16 variables. In \cite{Kr} (see also \cite{Hit}) one can find a complete description of solutions for the YBE in the case $n=2$.

V.~G.~Drinfeld \cite{Drinfeld} suggested studying set-theoretical solutions, i.e. solutions for which the vector space $V$ is  spanned by a set $X$, and $R$ is the linear operator induced by a permutation $R: X \times X \to  X \times X$. In this case, we say that $(X,R)$ is a set-theoretical solution of the Yang--Baxter equation or simply a solution of the YBE.  Set-theoretical solutions have connections for example with groups of I-type, Bieberbach groups, bijective 1-cocycles, Garside theory, etc. Involutive set-theoretical solutions for the YBE were studied in \cite{ESS}.

It is easy to see that for any $X$ the map $P(x, y) = (y, x)$ gives a set-theoretical solution of the YBE. On the other side, if $R$ is a solution of the YBE, then the map $S = P R$ satisfies the braid relation
$$
(S \times \id) (\id \times S) (S \times \id) = (\id \times S) (S \times \id) (\id \times S).
$$
The braid relation is a relation in the braid group $B_n$. Topologically, the braid relation is simply the third Reidemeister move of planar diagrams of links. 
In 1980s, D.~Joyce \cite{J} and S.~V.~Matveev~\cite{Mat} introduced quandles as invariants of knots and links, and proved that any quandle gives an elementary set-theoretical solution to the braid equation. 
A lot of papers is dedicated to set-theoretical solution of the Yang-Baxter equation and the braid equation (see, for example \cite{Bukh, D, LYZ00}).

In \cite{PT} solutions of a parametric YBE were constructed. In \cite{BIKP} there was suggested a method to construct  linear parametric solutions to the YBE from non-linear Darboux transformations of some Lax operators.

In 2006, W. Rump introduced~\cite{Rump2,Rump} braces to study involutive set-theoretical solutions to the Yang--Baxter equation, although the same algebraic objects were already considered by A.~G.~Kurosh~\cite{Kurosh} in 1970s. In~2014, this notion was reformulated by F.~Ced\'{o}, E.~Jespers, and J.~Okni\'{n}ski in~\cite{CedoJespersOkninski}. In 2017, L. Guarnieri and L. Vendramin defined~\cite{GV2017} skew left braces which give non-involutive solutions to the Yang--Baxter equation.

If $R(x, y) = (\sigma_y(x), \tau_x(y))$ for $x, y \in X$, is a solution of the YBE, then it is said to be non-degenerate if $\sigma_x$ and $\tau_x$ are invertible for all $x \in X$. It is said to be  square-free if $R(x,x) = (x,x)$ for all $x \in X$. 
The solution $R(x, y) = (\sigma_y(x), \tau_x(y))$ defines a 2-groupoid  $(X; \cdot, *)$ that is a set with two binary algebraic operations $\cdot, * : X \times X \to X$ that are $x \cdot y = \sigma_y(x)$ and $y * x = \tau_x(y)$. If $\sigma_y = \id$ for all $y \in X$ or $\tau_x = \id$ for all $x \in X$, then the solution $(X, R)$ is called by elementary solution. Any elementary solution defines a  self-distributive groupoid on $X$, if the elementary solution is non-degenerate, then this groupoid is a rack. On the other side, any self-distributive groupoid gives an elementary solution for the YBE and any rack gives a non-degenerate solution for the YBE.  Any quandle defines a non-degenerate square-free solution of the YBE.

The $3$-simplex equation or the tetrahedron  equation (TE)  has been introduced by A.~B.~Zamolodchikov \cite{Z, Z-1} as a 3-dimensional generalization of the YBE. 
The  TE is a non-linear relation in $V^{\otimes 6}$ on an endomorphism $R : V^{\otimes 3} \to V^{\otimes 3}$:
$$
	R_{123} R_{145} R_{246} R_{356} = R_{356} R_{246} R_{145} R_{123},
$$
where indexes denote the copies of $V$ in the tensor product on which the operator acts as $R$ on all other copies on acts as identity map. 
The tetrahedron equation is popular in electric  nets. In particular, the famous transformation `star-triangle' gives a rational solution to the tetrahedron equation. Also, the tetrahedron equation is used in the theory of 2-dimensional knots  in a 4-sphere \cite{KST}. It is some analogue of the third Reidemeister move in classical knot theory. 
The TE arrives in a problem of coloring 2-faces of a 4-dimensional cube by elements of some set $X$ (see \cite{KST-1}).

Many solutions of the TE have been already found. For example, J. Hietarinta \cite{HLinear} considered the set of integers $\mathbb{Z}_d$ modulo $d$ as $X$  and studied linear and affine maps $X^n \to X^n$ which give solutions to the $n$-simplex equation for $n = 2, 3, 4$. In \cite{KNPT} linear solutions for the TE are described.

The $n$-simplex equation is a generalization of YBE and TE. It was introduced in  \cite{BS} and studied in \cite{BV}, \cite{CS}, \cite{HLinear}, \cite{DK}, \cite{K}, \cite{S}. In the present paper  we suggest some general approaches to constructing solutions of $n$-simplex  equations, describe some types of solutions, introduce an operation which under some conditions allows us to construct a solution of  $(n+m+k)$-SE from solution of  $(n+k)$-SE and  $(m+k)$-SE. 
We consider the tropicalization of rational solutions and discuss a way to generalize it. We prove that if a group $G$ is an extension of a group $H$ by a group $K$, then  we can find a solution of the $n$-SE on $G$ from solutions of this equation on $H$ and on $K$. Also, we find solutions of the parametric Yang-Baxter equation on $H$ with parameters in $K$. For studying solutions of the 3-simplex equation we introduce algebraic systems with ternary operations and give examples of these systems which gives solutions of the $3$-SE. 
We find all elementary verbal solutions of the  $3$-SE  on free groups.\\

The paper is organized as follows. In Section \ref{prelim}, we recall  connection between set-theoretical solutions of the Yang--Baxter equation and some algebraic systems (self-distri\-bu\-tive groupoid, rack, quandle). We formulate conditions under which an algebraic system with two binary operations gives a set-theoretical solution of the YBE. We recall definitions and some properties of braid groups and virtual braid groups.

In Section \ref{NSim} we introduce a general form of the $n$-simplex equation, describe an algorithm that gives a possibility to write the $n$-simplex equation for arbitrary $n \geq 2$ and formulate results on general solutions. We also give a definition of the classical {\SE} that is an analogous to the classical YBE.

It is known that if $(X, R)$ is a solution of the YBE and $P_{12}(x, y) = (y, x)$ is the transposition, then the map $P_{12} R P_{12}$ satisfies the YBE. We prove in Proposition \ref{invsymm} a similar fact for any $n$-SE, constructing some permutation $P$ such that the map $P R P$ satisfies the $n$-SE.

For non-negative integer $k$ we introduce $k$-amalgam of solutions of the $(n+k)$-SE and $(m+k)$-SE and find  conditions (see Theorem \ref{const}) under which $k$-amalgam is  a solution of the  \SE[(n+m+k)]. As a corollary of this general construction we study composition of linear solutions.
In Proposition \ref{sumlin} we prove that $ 0 $-amalgam of linear solutions of {\SE} and \SE[m] gives a solution of the \SE[(n+m)]. If for these linear solutions $1$-amalgam is defined, then it is a linear solution of the \SE[(n+m-1)]. We also prove that if a solution of the  {\SE} has a fixed point, then a specialization of this solution gives a solution of the \SE[(n-1)].

In Section \ref{Trop} we consider some class of rational solutions of the {\SE} and give a definition of their tropicalization. We show that the tropicalization of such rational solutions gives piecewise linear solutions of the {\SE}. In particular, we construct a linear solution of the TE using the famous electric solution. Further we generalize the tropicalization on other algebraic systems. As an example, we show that using the  homomorphism $h : \mathbb{R} \to \mathcal{D}'[\mathbb{R}^m]$ from the field $\mathbb{R}$ of real numbers with standard operations of addition and multiplication	to the linear space $\mathcal{D}'[\mathbb{R}^m]$ of generalized functions with compact support on the space $\mathbb{R}^m$ with the standard operation of addition and the convolution $*$ as the multiplication, one can construct a solution $(\mathcal{D}'[\mathbb{R}^m], R^h)$ of the {\SE} from a rational solution $(\mathbb{R}, R)$ of the {\SE}.
	
In Section \ref{GrExt} we consider a group $G$ with a normal subgroup $H$ and the quotient $G / H \cong  K$. We prove that if  a structure of a self-distributive groupoid on $G$ is defined, then, under some conditions on $K$, there is a solution of the YBE on $H$ with parameters in $K$. In particular, it is true (see Corollary \ref{trext}) if $K$ is the trivial self-distributive groupoid. 
Further these results are extended to general {\SE}. In Subsection \ref{lin} using the known representations of the  virtual braid group we construct (see Theorem \ref{psol}) two types solutions of the Yang--Baxter equations on a group with  parameters  in an abelian group.

Further we consider a group $G$ that is an extension of a group $H$ by a group $K$.
In Proposition \ref{trext} we prove that by the solutions $(H, R)$ and $(K, T)$ of the {\SE} it is possible to construct a solution of the {\SE} on $G$.
	
In Section \ref{Inlim} we define a category of solutions of the {\SE} and prove in Proposition \ref{lim} that the limit of solutions of the {\SE} is also a solution of the {\SE}.
 As a consequence we get (see Corollary \ref{adic}) that if we have solutions $(\mathbb{Z}_{p^k}, R_k)$, $k = 1, 2, \ldots$, of the {\SE}, then under some conditions we can define a solution of the {\SE} on the set of the $p$-adic integer numbers.	


In Section \ref{algsust} by analogy with racks, bi-racks, skew braces, which were introduced for the studying of the Yang-Baxter equation, we introduce  ternar and 3-ternoid, that are algebraic systems with one and three ternary operations, respectively. We show how they are connected with solutions of the tetrahedral equation (see Proposition \ref{p7.1} and its corollaries). Also, we introduce some algebraic systems with four binary operations and show in Propositions    \ref{2-sol} and   \ref{1-sol} how it is possible to find elementary solutions for the tetrahedron equation.

In Section \ref{verb} we defined verbal solution. These are some generalizations of verbal solutions of the YBE, which one can define on a group  with structure  of conjugacy  quandle or core quandles. We describe all elementary verbal solutions of the TE, which are possible to define on free non-abelian groups.

In the paper we formulate some open questions for further research.

Throughout the paper we use standard notations. If $G$ is a group, and $a,b\in G$, then we denote by $a^b=b^{-1}ab$ the conjugate of $a$ by $b$, and by $[a, b]=a^{-1} b^{-1} a b$ the commutator of $a$ and $b$. 
Also, we shall assume that  compositions acting from the right to the left, i.e. if $f, g : X \to X$, then their composition $f g$ acts by the rule 
$(f g)(x)=f(g(x))$ for $x \in X$ in which it is defined.


\bigskip


\section{Preliminaries and some basic results}\label{prelim}
In this section we give necessary preliminaries.  We recall examples of solutions to the YBE, describe  a connection of solutions with some groupoids (racks, quandles) and recall definitions of the braid group and the virtual braid group.

\subsection{Yang--Baxter equations and racks}
The following lemma is evident.
\begin{lemma} \label{YBEF}
Let $X$ be a non-empty set and $f, g : X^2 \to X$ are two maps. The map $R : X^2 \to X^2$, $R(x, y) = (f(x, y), g(x, y))$ is a solution of the YBE if and only if for any $x, y, z \in X$ the following identities hold
\begin{eqnarray*}
f\left( f(x, y), z  \right) &=& f\left( f(x, g(y, z)),   f(y,  z) \right),\\
f\left( g(x,y),  g(f(x,y), z) \right) &=& g\left( f(x, g(y, z)),   f(y,  z) \right),\\
g\left( g(x, y),   g(f(x, y), z) \right) &=& g\left( x,  g(y, z)  \right).
\end{eqnarray*}
\end{lemma}

\begin{example}
\begin{enumerate}
\item Let $A$ be an abelian group and $a, b \in A$, then
$$
R(x, y) = (x+a, y+b),~~x, y \in A,
$$
is a solution of the YBE on $A$.

\item For arbitrary group $G$ the map
$$
R(x, y) = (xy^2, x y^{-1} x^{-1}),~~x, y \in G,
$$
is a solution of the YBE on $G$. 

\item Let $G$ be  a group and $\varphi \in \Aut(G)$ be its automorphism, then the map
$$
R(x, y) = (\varphi(x), x y \varphi(x)^{-1}),~~x, y \in G,
$$
is a solution of the YBE on $G$. It is easy to see that if the map $\varphi(x) = x y^2$ is an automorphism of $G$, we get the solution from the previous example.
\end{enumerate}
\end{example}

A groupoid is a non-empty set with one binary algebraic operation,
if $k$ is a natural number, $k>1$, then  a $k$-groupoid is a non-empty set with $k$ binary algebraic operations. 
Suppose that $(X, R)$ is a set-theoretical solution of the YBE. 
Writing $R(x, y) = (\sigma_y(x), \tau_x(y))$ for $x, y \in X$, we say that the solution $R$ is {\it non-degenerate} if $\sigma_x$ and $\tau_x$ are invertible for all $x \in X$, we says that the solution $R$ is {\it square-free} if $R(x,x) = (x,x)$ for all $x \in X$.

A {\it quandle} is a groupoid $Q$ with a binary operation $(x,y) \mapsto x * y$ satisfying the following axioms:
\begin{enumerate}
\item[(Q1)] $x*x=x$ for all $x \in Q$,
\item[(Q2)] for any $x,y \in Q$ there exists a unique $z \in Q$ such that $x=z*y$,
\item[(Q3)] $(x*y)*z=(x*z) * (y*z)$ for all $x,y,z \in Q$.
\end{enumerate}

If a groupoid satisfies (Q3) it is called a self-distributive groupoid, if it satisfies 
 (Q2) and (Q3) it is called a {\it rack}. Many interesting examples of quandles come from groups showing deep connection with group theory.

\begin{example}
\begin{enumerate}
\item If $G$ is a group, then the binary operation $a*b= b^{-1} a b$ turns $G$ into the quandle $\Conj(G)$ called the {\it conjugation quandle} of $G$.
\item A group $G$ with the binary operation $a*b= b a^{-1} b$ turns the set $G$ into the quandle $\Core(G)$ called the {\it core quandle} of $G$. In particular, if $G= \mathbb{Z}_n$, the cyclic group of order $n$, then it is called the {\it dihedral quandle} and denoted by $\R_n$.
\item Let $G$ be a group and $\phi \in \Aut(G)$. Then the set $G$ with binary operation $a * b = \phi(ab^{-1})b$ forms a quandle $\Alex(G,\phi)$ referred as the  {\it generalized Alexander quandle} of $G$ with respect to $\phi$.
\end{enumerate}
\end{example}

\medskip

A quandle  $Q$ is called {\it trivial} if $x*y=x$ for all $x, y \in Q$.  Unlike groups, a trivial quandle can have arbitrary number of elements. We denote the $n$-element trivial quandle by $\T_n$ and an arbitrary trivial quandle by $\T$.

\medskip

Notice that the axioms (Q2) and (Q3) are equivalent to the map $S_x: Q \to Q$ given by $$S_x(y)=y*x$$ being an automorphism of $Q$ for each $x \in Q$. These automorphisms are called {\it inner automorphisms}, and the group generated by all such automorphisms is denoted by $\Inn(X)$. A quandle is said to be {\it connected}  if it admits a transitive action by its group of inner automorphisms. For example, dihedral quandles of odd order are connected, whereas that of even order are disconnected. A quandle $X$ is called {\it involutive} if $S_x^2 = \id_Q$ for each $x \in Q$. For example, all core quandles are involutive.

For a solution $(X, R)$, $R(x, y) = (\sigma_y(x), \tau_x(y))$,  we can define a 2-groupoid  $(X; \cdot, *)$ with the operations   $x \cdot y = \sigma_y(x)$ and $y * x = \tau_x(y)$. If $\sigma_y = \id$ for all $y \in X$ or $\tau_x = \id$ for all $x \in X$, then 
the solution $(X, R)$ is called by {\it elementary solution}. Any elementary solution defines a structure of self-distributive groupoid on $X$, if the elementary solution is non-degenerate, then this groupoid is a rack. On the other side, any self-distributive groupoid gives an elementary solution to the YBE and any rack gives a non-degenerate solution to the YBE, any quandle defines a non-degenerate square-free solution of the YBE.

Connections  of solutions with 2-groupoids gives the next lemma, which follows from Lemma \ref{YBEF}.

\begin{lemma} \label{YBT}
Let $(X, \cdot, *)$ be a 2-groupoid and   $R : X \times X \to X \times X$ given by $R(x, y) = (x \cdot y, y * x)$ for $x, y \in X$. Then the following hold:
\begin{enumerate}
\item The pair $(X, R)$ is a set-theoretical solution of the Yang--Baxter equation if and only if the equalities
\begin{eqnarray*}
(x \cdot y) \cdot z  &=& (x \cdot (z * y)) \cdot (y \cdot  z),\\
(y * x) \cdot (z * (x \cdot y))  &=&  (y \cdot z)  *(x \cdot (z * y)),\\
(z*(x \cdot y)) * (y * x)  &=&  (z* y ) * x,
\end{eqnarray*}
hold for all $x, y, z \in X$. 
\item If $x \cdot y = y$ for all $x, y \in X$, then the  pair $(X, R)$  is a set-theoretic solution of the Yang--Baxter equation if and only if the operation $*$ is (right) self-distributive, i.e.
\begin{equation*}\label{rack-eqn}
(z * x) * (y * x) = (z*y) * x
\end{equation*}
 for all $x, y, z \in X$.
\item If $y * x = y$ for all $x, y \in X$, then the  pair $(X, R)$  is a set-theoretic solution of the Yang--Baxter equation if and only if the operation $\cdot$ is (right) self-distributive, i.e.
\begin{equation*}\label{rack-eqn-1}
(z\cdot y) \cdot x = (z \cdot x) \cdot (y \cdot x) 
\end{equation*}
 for all $x, y, z \in X$.

\end{enumerate}
\end{lemma}

\begin{example}
Let us define on the set of real numbers $\mathbb{R}$ a binary operation 
$$
a * b = \max\{a, b\}, ~~~a, b \in \mathbb{R}.
$$ 
Then $(\mathbb{R}, *)$ is a self-distributive groupoid, which is not a rack, but any element $a \in \mathbb{R}$  is idempotent, that is $a * a = a$, and it is commutative groupoid. The map $R(x, y) = (x, y*x)$ gives a degenerate solution of the YBE on $\mathbb{R}$. If we take $\min$ instead of $\max$ we get a similar result. 
\end{example}

A. Soloviev~\cite[Theorem 2.3]{Sol} (see also L.~Guarnieri, L.~Vendramin~\cite[Proposition 3.7]{GV2017} and D.~Bachiller \cite[Proposition 5.2]{Bach}) proves that any non-degenerate solution of the YBE conjugates to an elementary solution.

\begin{proposition}
If $R(x,y) = (\sigma_y(x), \tau_x(y))$, $x, y \in X$, gives a~left non-degenerate (i.e. $\sigma_y$ is a bijection for any $y \in X$)
solution to the YBE, then it conjugates to an elementary  solution:
$$
R'(x,y) = (\sigma_x(\tau_{\sigma^{-1}_y(x)}(y)), y).
$$
If for all $a, b \in X$ there exists a unique $x \in X$ such that
\begin{equation}\label{cond}
\tau_{\sigma^{-1}_x(a)}(x) = \sigma^{-1}_a(b),
\end{equation}
then this solution is non-degenerate.
\end{proposition}

\subsection{Braid group and virtual braid group}

For $n\geq 2$ the braid group $B_n$ on $n$ strands is defined as
the group generated by $(n-1)$ elements
$\sigma_1,\sigma_2,\ldots,\sigma_{n-1}$
with defining relations
\begin{align}
 \label{eq1}\sigma_i\sigma_{i+1}\sigma_i& = \sigma_{i+1}\sigma_i\sigma_{i+1},&&i=1,2,\ldots,n-2,\\
\label{eq2}\sigma_i\sigma_j &= \sigma_j\sigma_i,&& |i-j|\geq 2. 
\end{align}

\medskip

The \emph{virtual braid group} $VB_n$ on $n$ strands, introduced in \cite{Ka}, is the group with the generators $\sigma_1,\sigma_2,\dots,\sigma_{n-1}$, $\rho_1,\rho_2,\dots,\rho_{n-1}$, which is defined by the relations (\ref{eq1})-(\ref{eq2}) and additional relations
 \begin{align}
\label{eq12} \rho_{i}  \rho_{i+1} \rho_{i} &= \rho_{i+1} \rho_{i} \rho_{i+1},&&  i=1,2,\ldots,n-2,\\
\label{eq13}\rho_{i}  \rho_{j} &= \rho_{j}  \rho_{i},&&  |i-j| \geq 2,   \\
\label{eq14}\rho_{i}^2 &= 1,&&  i=1,2,\ldots,n-1,\\ 
\label{eq15}\sigma_{i} \rho_{j} &= \rho_{j}  \sigma_{i},&& |i-j| \geq 1,\\
\label{eq16}\rho_{i} \rho_{i+1}  \sigma_{i} &= \sigma_{i+1}  \rho_{i}   \rho_{i+1},&&  i=1,2,\ldots,n-2. 
\end{align}
The elements $\sigma_1,\sigma_2,\dots,\sigma_{n-1}$ generate a subgroup that is isomorphic to the braid group $B_n$. Relations (\ref{eq12})-(\ref{eq14}) are defining relations of the symmetric group $S_n$, there, the elements $\rho_1,\rho_2,\dots,\rho_{n-1}$ generate  $S_n$ in $VB_n$, hence the virtual braid group $VB_n$ is generated by the braid group $B_n$ and the symmetric group $S_n$. Relations (\ref{eq15})-(\ref{eq16}) are called the mixed relations of the virtual braid group.


\bigskip

\section{Quantum and classical $n$-simplex equations } \label{NSim}

\subsection{Construction of quantum $n$-simplex equations}
To write  a quantum $n$-simplex equation (\SE) for any natural $n$, we give a geometric interpretation of the YBE, which is the 2-simplex equation. For simplicity, we will call  a quantum $n$-simplex equation by a $n$-simplex equation. Suppose that we have three straight lines $l_1$, $l_2$, and $l_3$ on the plane $\mathbb{R}^2$. The line $l_1$ intersects with $l_2$ in the point $R_{12}$, with the line $l_3$ in the point $R_{13}$, and the line $l_2$ intersects with $l_3$ in the point $R_{23}$. We assume that all points $R_{12}$, $R_{13}$, and $R_{23}$ are different  and are vertices of a non-degenerate triangle (2-simplex). See Figure 1.

\begin{figure}[h]
\noindent\centering{\includegraphics[height=0.5 \textwidth]{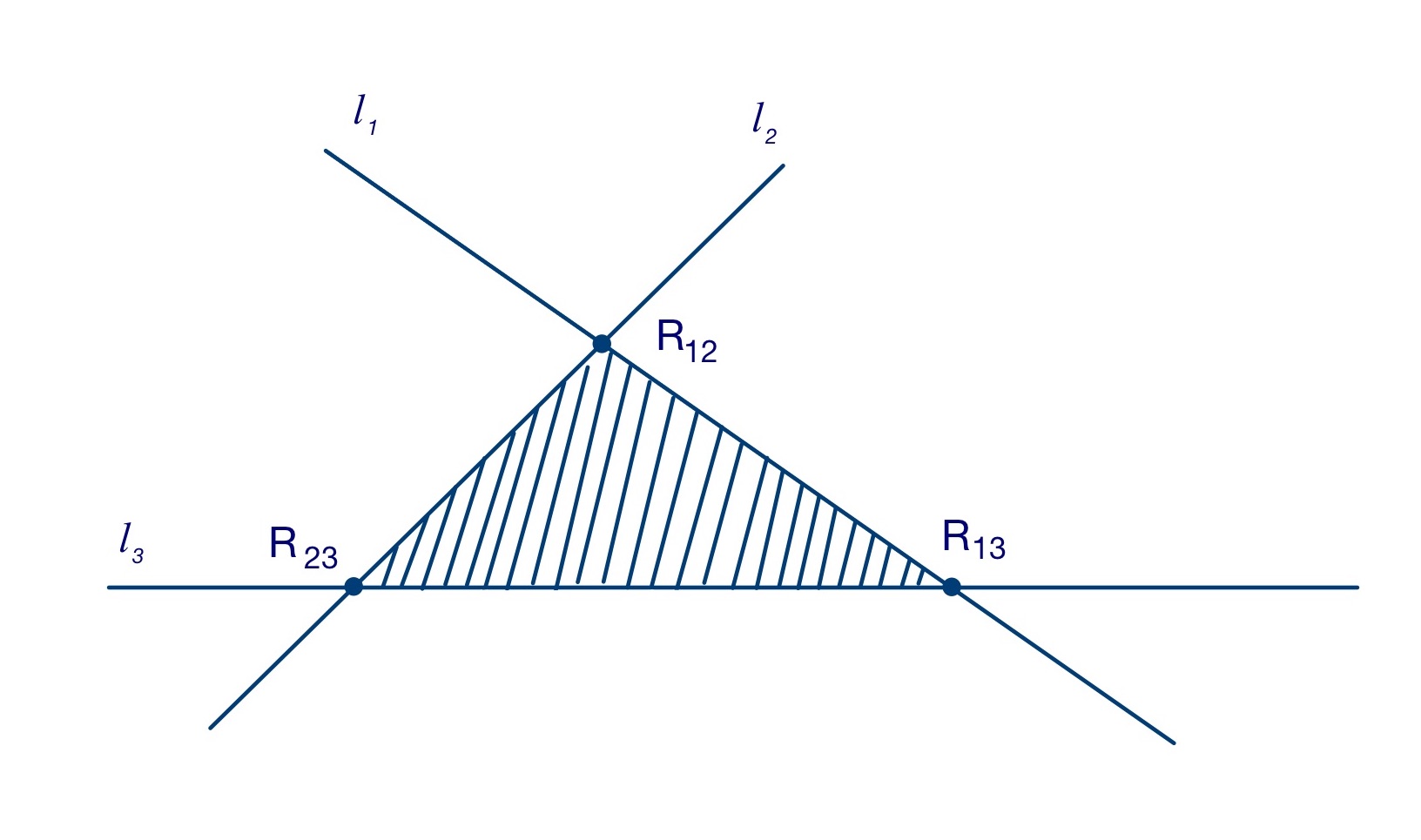}}
\caption{Geometric interpretation of the YBE}
\end{figure}

Using the lexicographical order, we introduce an order on the vertices,  
$$
R_{12} < R_{13} < R_{23}.
$$
Then the YBE is the equality of words, where the first one  is a word which we get by going around  the vertices in the increasing order and the second word is a word which we get by going around the vertices in the decreasing order, i.e.
$$
R_{12} R_{13} R_{23} = R_{23} R_{13} R_{12}.
$$
To get the tetrahedron equation ($3$-SE) we  increment the indices of all lines by 3 and get the triangle with the vertices $R_{45}$, $R_{46}$, and $R_{56}$. Further, embed our plane $\mathbb{R}^2$ into a 3-space $\mathbb{R}^3 = \mathbb{R} \times \mathbb{R}^2$, take a vertex $R_{123}$, which does not lie in $\mathbb{R}^2$ and construct a straight line $l_1$, which connects $R_{123}$ with the first vertex $R_{45}$ of the original triangle. Analogously, construct a straight line $l_2$, which connects $R_{123}$ with the second vertex, $R_{46}$ and construct a straight line $l_3$, which connects $R_{123}$ with the third vertex, $R_{56}$. We construct a tetrahedron with the vertices $R_{123}$, 
$R_{145}$, $R_{246}$, and $R_{356}$. See Figure 2.

\begin{figure}[h]
\noindent\centering{\includegraphics[height=0.7 \textwidth]{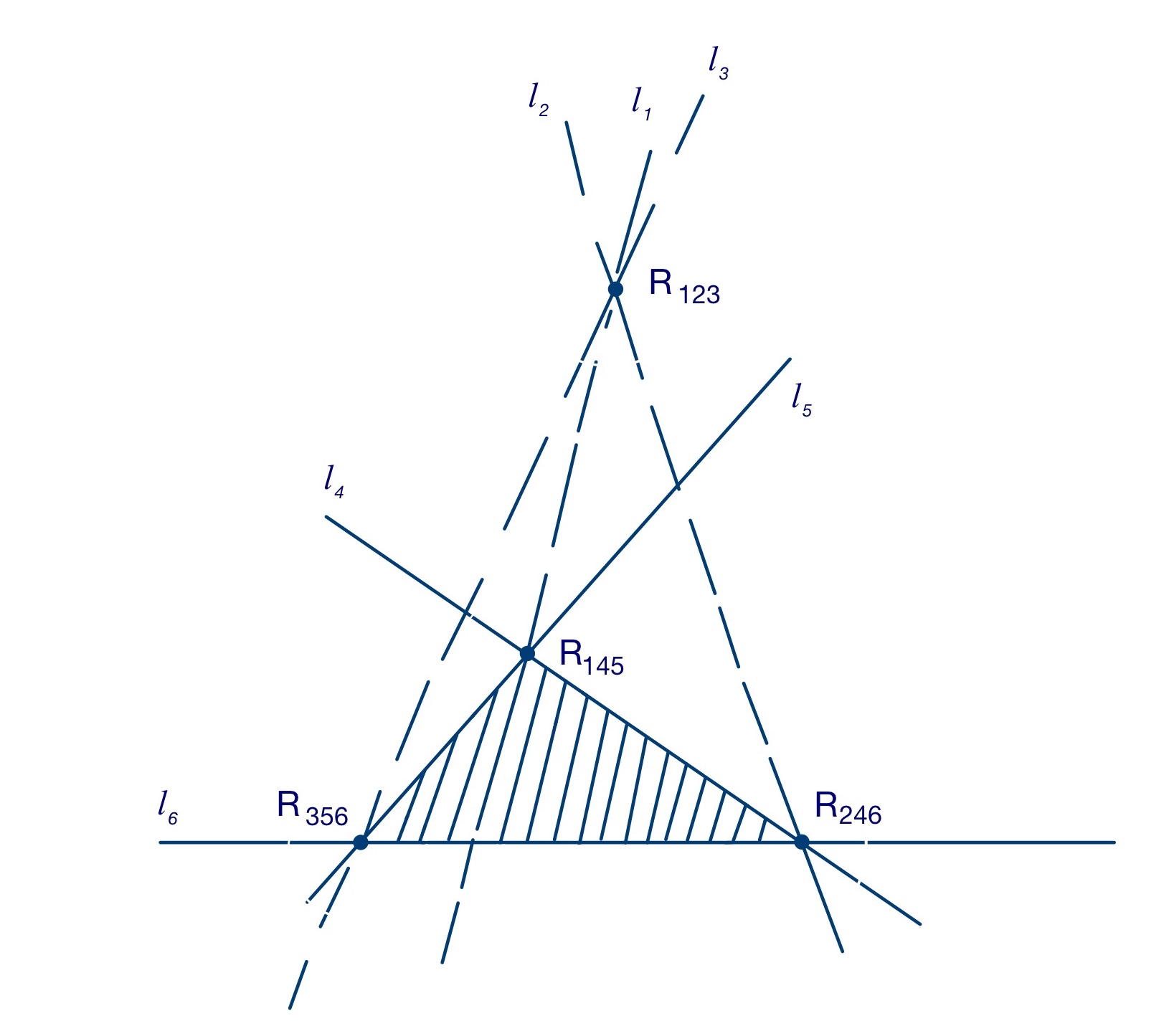}}
\caption{Geometric interpretation of TE}
\end{figure}

Then  the tetrahedron equation (TE) is the equality of words, where the first one  is a word which we get by going around  the vertices of the tetrahedron  in the increasing order and the second word is a word which we get by going around  the vertices in the decreasing order, i.e.
$$
	R_{123} R_{145} R_{246} R_{356} = R_{356} R_{246} R_{145} R_{123}.
$$

Using the same approach one can write the  4-simplex equation,
$$
	R_{1234} R_{1567} R_{2589} R_{368,10} R_{479,10} = R_{479,10} R_{368,10}  R_{2589} R_{1567} R_{1234}.
$$

In general case, the left side and the right side of the {\SE} are words of length $n+1$.
Notice  that $n+1$ is the number of vertices of the $n$-simplex.
Suppose that we have the {\SE} for some $n\geq 3$. We will write this equation in the form 
$$
	R_{\overline{1}}R_{\overline{2}} \cdots R_{\overline{n+1}}
	=
	R_{\overline{n+1}} \cdots R_{\overline{2}} R_{\overline{1}},
$$
where $ \overline{k} = (k_1, k_2, \ldots, k_{n+1}) \in \mathbb{N}^{n+1} $ is a multi-index. We want to construct the $(n+1)$-SE.
To do it, define an operation $ s_n : \mathbb{N} \to \mathbb{N} $, $ s_n(k) = k + n + 1 $, and extend it to the multi-indices by the rule
$$
	s_n(\overline{k}) = (s_n(k_1), s_n(k_2), \ldots, s_n(k_{n+1})) \in \mathbb{N}^{n+1}.
$$
Using these notations we can write the {\SE[(n+1)]} as
$$
	R_{1,2,\ldots,n+1}	R_{1, s_n(\overline{1})} R_{2, s_n(\overline{2})} \cdots R_{n+1, s_n(\overline{n+1})}
	=
	R_{n+1, s_n(\overline{n+1})} \cdots R_{2,s_n(\overline{2})} R_{1,s_n(\overline{1})} R_{1,2,\ldots,n+1}.
$$

The multi-indices in the {\SE} can be regarded as rows of the {\it multi-indices matrix} $MI_n$ that is a $(n+1) \times n$ matrix satisfying the recurrence relation
$$
MI_{n} = 
\begin{pmatrix}
 1  & & 2 & 3 & \cdots & n  \\
  \hline
1 &\vline & &  &  & \\
2 &\vline &  & &   &\\
\vdots &\vline &  & MI_{n - 1} + (n) &  & \\
n-1 &\vline &  &  &   &\\
n &\vline &  &  &  & \\
\end{pmatrix}
$$
where $(n)$ is the $ n $ by $ n-1 $ matrix in which all elements are equal to $n$.
The matrix $ MI_n $ could be written explicitly as 
$$
	\begin{pmatrix}
	1 & 2 & 3 & \cdots & n \\
	1 & n + 1 & n + 2 & \cdots & 2n - 1 \\
	2 & n + 1 & 2n & \cdots & 3n - 3 \\
	\vdots & \vdots & \vdots & \ddots &\vdots \\
	n-1 & 2n - 2 & 3n - 4 & \cdots & \frac{n(n+1)}{2}\\
	n & 2n - 1 & 3n - 3 & \cdots & \frac{n(n+1)}{2}
	\end{pmatrix}
$$

A linear solution of the {\SE} is a linear map $R: V^{\otimes n} \to V^{\otimes n}$ such that the next equality
$$
	R_{\overline{1}}R_{\overline{2}} \cdots R_{\overline{n+1}}
	=
	R_{\overline{n+1}} \cdots R_{\overline{2}} R_{\overline{1}},
$$
of two linear maps $V^{\otimes N} \to V^{\otimes N}$, where $N = n(n+1)/2$, holds. In this equality the map $R_{\overline{k}} : V^{\otimes N} \to V^{\otimes N}$
acts as $R$ on the copies of $V$ with the numbers $\overline{k}$ and as identity map on copies with other numbers.
Notice that $n+1$ is the number of vertices of the $n$-simplex, and $N$ is the number of its edges.

A set-theoretic solution of the {\SE} on a set $X$ is a map $R: X^{n} \to X^{n}$ that is satisfies the {\SE}.

For completeness,  we define the $1$-simplex equation  $R_1 R_2 = R_2 R_1$ as equality of two maps $X \times X \to X \times X$. In this case $R$ is a map $R: X \to X$. Since,
$$
R_1(x, y) = (R(x), y),~~R_2(x, y) = (x, R(y)),~~R_1 R_2 (x, y) = R_2 R_1 (x, y) = (R(x), R(y)),
$$
any map $R: X \to X$ gives a solution to the $1$-simplex equation.

\subsection{Classical $n$-simplex equation}

The equation 
$$
R_{12} R_{13} R_{23} = R_{23}  R_{13} R_{12}
$$
 also is called the Quantum Yang-Baxter Equation. If $(V, R)$ is its linear solution and the linear map $R$  can be presented in the form
$$
R = 1 +  \hbar r + O(\hbar^2),
$$
then $r\colon V\to V$ gives a solution to the Classical Yang-Baxter Equation,
$$
[r_{12}, r_{13}] + [r_{12}, r_{23}] + [r_{13}, r_{23}] = 0,
$$
where $[a, b] = ab - ba$. The \SE \, which we considered before  are quantum \SE. It is interesting to understand what is a classical version of the \SE?  The proof of the following proposition is strait calculations.

\begin{proposition}
If $V$ is a vector space and $R\colon V^n\to V^n$ is a linear map that is  a  solution of the quantum \SE, and there is a map $r \colon V^n \to V^n$ such that 
$$
R = 1 +  \hbar r + O(\hbar^2),
$$
then $r$ satisfied the equations
$$
\sum_{\overline{1} \leq \overline{i} <  \overline{j} \leq \overline{n+1}}
[ r_{\overline{i}},  r_{\overline{j}}] = 0,
$$
where we use the lexicographic order on the  multi-indices. This  equation  is said to be  the classical \SE.
\end{proposition}

\bigskip

\section{Solutions of arbitrary \SE}

Further we will consider only set-theoretic solutions and call them simply by solutions. We also will work with partial functions $R :  X^n \to X^n$ defined only on some subset $D \subset X^n$.
In these cases we will say that $(X, R)$ is a solution of the {\SE} when the {\SE} holds for every point $\bar{x} \in X^N$
where it has sense, i.e. where left and right parts of the equation are defined.

The next proposition is a generalization of well-known result for the YBE.

\begin{proposition}
	Let  $R : X^n \to X^n $ be a solution of the {\SE}. 
\begin{enumerate}	
\item	If $R$ is invertible, then its inverse $R^{-1}$	is also a solution of the {\SE}.
\item 	If $\varphi _1, \dotss , \varphi _n$ are pairwise commuting maps from $X$ to $X$, then the map $R : X^n \to X^n $ defined as
	$$
		R(x_1, \dotss , x_n) = (\varphi _1 (x_1), \dotss , \varphi _n (x_n)),
	$$
	is a solution of the {\SE}.
\item	If $ \varphi \in Sym(X) $ is an arbitrary bijection of the set $X$ onto itself, then the composition 
$$
 \varphi^{\times n}  \, R  \, (\varphi^{-1})^{\times n}
 $$ is a solution of the {\SE}.
\end{enumerate}
\end{proposition}

\begin{proof}
(1)	If we take inverse on both sides of {\SE}, we get that $R^{-1}$ gives a solution.

(2) 	It is easy to see that all functions $ R_{\overline{\imath}} $ derived from such $ R $ are pairwise commutative.
	Thus the left side of the {\SE} could be reordered to obtain the right side of the equation.

(3)	 Put $\widetilde R = \varphi^{\times n} \, R \, (\varphi^{-1})^{\times n}$. Then 
$$
\widetilde{R}_{\bar k} = \varphi^{\times n}_{\bar k} \, R_{\bar k} \, ({\varphi^{-1}})^{\times n}_{\bar k} = \varphi^{\times N} \, R_{\bar k} \, ({\varphi^{-1}})^{\times N}$$
The equation
$$
	\widetilde{R}_{\overline{1}}\widetilde R_{\overline{2}} \cdots  \widetilde R_{\overline{n+1}}
	=
	\widetilde R_{\overline{n+1}} \cdots \widetilde R_{\overline{2}} \widetilde R_{\overline{1}},
$$
is equivalent to the equation

$$\begin{gathered}
	\left( \varphi^{\times N} \, R_{\bar 1} \, ({\varphi^{-1}})^{\times N} \right) \, \left( \varphi^{\times N} \, R_{\bar 2} \, ({\varphi^{-1}})^{\times N}\right)  \cdots  \left( \varphi^{\times N} \, R_{\overline {n+1}} \, ({\varphi^{-1}})^{\times N}\right) 
	=\\
	= \left( \varphi^{\times N} \, R_{\overline{n+1}} \, ({\varphi^{-1}})^{\times N} \right)  \cdots \left(  \varphi^{\times N} \, R_{\bar 2} \, ({\varphi^{-1}})^{\times N} \right)  \left(  \varphi^{\times N} \, R_{\bar 1} \, ({\varphi^{-1}})^{\times N} \right).
	\end{gathered}
$$
After cancellations we obtain an equation
 $$
	\varphi^{\times N} {R}_{\overline{1}} R_{\overline{2}} \cdots  R_{\overline{n+1}} ({\varphi^{-1}})^{\times N} 
	=
	\varphi^{\times N} R_{\overline{n+1}} \cdots  R_{\overline{2}}  R_{\overline{1}}\ ({\varphi^{-1}})^{\times N}. 
$$
Since $\varphi^N$ is a bijection, the obtained equation holds if and only if the initial equation holds.
\end{proof}

The following proposition is a generalization of  Proposition 2.2 from \cite{Kassotakis}.
\begin{proposition} \label{twist} 
If $ R $ is a solution of a {\SE} and $ \varphi \in Sym(X) $ is such that
	$$ (\varphi^{-1})^{\times n} \, R \, \varphi^{\times n} = R $$
then the map
	$$ \tilde R = (\varphi \times \id \times \varphi \times \ldots) \, R \, (\id \times \varphi^{-1} \times \id \times \ldots),
	$$
where $(\varphi \times \id \times \varphi \times \ldots)$ acts as $\varphi$ on odd components and as $\id$ on even,
$(\id \times \varphi^{-1} \times \id \times \ldots)$ acts as $\varphi^{-1}$ on even components and as $\id$ on odd is also  a solution of the {\SE}.
\end{proposition}
\begin{proof}
Note that $$(\varphi \times \id \times \varphi \times \ldots) =  (\id \times \varphi^{-1} \times \id \times \ldots) \, \varphi^{\times n}, $$

$$(\id \times \varphi^{-1} \times \id \times \ldots ) = (\varphi^{-1})^{\times n} \, (\varphi^{} \times \id \times \varphi^{} \times \ldots)$$
and  
$$ \begin{gathered}
R_A := \tilde R =
(\varphi \times \id \times \varphi \times \ldots)  R   (\id \times \varphi^{-1} \times \id \times \ldots) = \\ =
(\varphi \times \id \times \varphi \times \ldots)  \left ( (\varphi^{-1})^{\times n}   R   \varphi^{\times n} \right )   (\id \times \varphi^{-1} \times \id \times \ldots) = \\ =
(\id \times \varphi^{-1} \times \id \times \ldots)   R   (\varphi^{} \times \id \times \varphi^{} \times \ldots) =: R_B.
\end{gathered}$$

Writing the {\SE} one can observe that every index which first occurrence was in multi-index $\overline{i}$ on the position $j$ will appear again only in multi-index $\overline{j+1}$ on the position $i$. It means that if we replace every $R_{\overline{i}}$ where $i$ is even with $(R_{A})_{\overline{i}}$ and every $R_{\overline{i}}$ where $i$ is odd with $(R_{B})_{\overline{i}}$ then for every $i$ and every $j \geqslant i$ the component 
with index $i_j$ would be trivially acted upon between $R_{\overline{i}}$ and $R_{\overline{j+1}}$. Indeed if $i$ and $j$ are even or odd simultaneously then $i_j$ between $R_{\overline{i}}$ and $R_{\overline{j+1}}$ is acted upon by $id \circ id$, and if they have distinct parity then it is acted upon by either $\varphi\circ\varphi^{-1}$ or $\varphi^{-1}\circ\varphi$.
The final observation of the proof is that since every index appears in precisely two different multi-indices the equation
$$
	(R_B)_{\overline{1}}(R_A)_{\overline{2}} \ldots (\tilde R)_{\overline{n+1}} =
	(\tilde R)_{\overline{n+1}}\ldots (R_A)_{\overline{2}}(R_B)_{\overline{1}}
$$
takes the form
$$
	\Phi R_{\overline{1}}R_{\overline{2}}...R_{\overline{n+1}} \Phi =
	\Phi R_{\overline{n+1}}...R_{\overline{2}}R_{\overline{1}} \Phi,
$$
where $\Phi $ is componentwise action of either $\varphi$, $\varphi^{-1}$ or $\id$.
\end{proof}

\medskip

It is easy to check that if $(X, R)$, $R(x, y) = (x, y*x)$, is an elementary solution of the YBE and $P(x, y) = (y, x)$ for all $x, y \in X$, then
$$
P R P(x, y) = (x*y, y)
$$
is also an elementary solution of the YBE. 

The proof of the next lemma is elementary, but we give it to further generalize on arbitrary $n$-simplex equations.

\begin{lemma} 
Suppose that $R(x, y) = (x \cdot y, y * x)$ is a solution of the YBE. Then $P R P(x, y) = (x * y, y \cdot x)$ is also a solution of the YBE. 
\end{lemma}

\begin{proof}
 It is easy to see that
$$
R_{13} = P_{23}R_{12}P_{23},~~R_{23} = P_{12} P_{23}R_{12}P_{23}P_{12},
$$
where $P_{ij} : X^3 \to X^3$ permutes the $i$-th and $j$-th components. Then the YBE has the form
$$
R_{12} \cdot R_{12}^{P_{23}} \cdot R_{12}^{P_{23}P_{12}} = R_{12}^{P_{23}P_{12}} \cdot R_{12}^{P_{23}} \cdot R_{12}.
$$

Denote by $\widetilde{R} = P R P$. Then $R = P \widetilde{R} P$ and from the YBE 
$$
\widetilde{R}_{12}^{P_{12}} \cdot  \widetilde{R}_{12}^{P_{12} P_{23}} \cdot \widetilde{R}_{12}^{P_{12}P_{23}P_{12}} = 
\widetilde{R}_{12}^{P_{12}P_{23}P_{12}} \cdot  \widetilde{R}_{12}^{P_{12} P_{23}} \cdot \widetilde{R}_{12}^{P_{12}}.
$$
Conjugating both sides of this equation by $P_{12} P_{23}P_{12} = P_{13}$ we get
$$
\widetilde{R}_{12}^{P_{23} P_{12}} \cdot  \widetilde{R}_{12}^{P_{23}} \cdot \widetilde{R}_{12}= 
\widetilde{R}_{12} \cdot  \widetilde{R}_{12}^{P_{23}} \cdot \widetilde{R}_{12}^{P_{23} P_{12}}.
$$
That is equivalent to
$$
\widetilde{R}_{12} \widetilde{R}_{12} \widetilde{R}_{12}= 
\widetilde{R}_{12}   \widetilde{R}_{12}  \widetilde{R}_{12}.
$$
It means  that $\widetilde{R} = P R P$ satisfied the YBE.

\end{proof}

To generalize this result, let us consider a solution $(X, R)$ of the {\SE}  and define a permutation $P \in Sym(X^n)$ by
$$
P = P_{1,n} P_{2,n-1} \ldots P_{[n/2]+1,n-[n/2]},
$$
where $[n/2]$ is the integer part of $n/2$ and for $n = 2m+1$, we assume that the last permutation $P_{[n/2]+1,n-[n/2]} = P_{m+1,m+1}$ is the identity. It is easy to see that $P^{-1} = P$. In these notations we can formulate the following proposition.

\begin{proposition} \label{invsymm}
	Let $ R : X^n \to X^n $ be a solution of the {\SE}. 
	Then $  \widetilde{R} = P R P $ is also a solution of the {\SE}.
\end{proposition}

\begin{proof} 
The \SE
$$
	R_{\overline{1}}R_{\overline{2}} \cdots R_{\overline{n+1}} 	= 	R_{\overline{n+1}} \cdots R_{\overline{2}} R_{\overline{1}}
$$
can be presented in the form
$$
	R_{\overline{1}} R_{\overline{1}}^{Q_1} \cdots R_{\overline{1}}^{Q_n} 	= 	R_{\overline{1}}^{Q_n} \cdots R_{\overline{1}}^{Q_1} R_{\overline{1}},
$$
where $Q_1, \ldots, Q_n \in Sym(X^n)$ are some permutations. Since $R = P  \widetilde{R} P$, we have
$$
	\widetilde{R}_{\overline{1}}^{P} \widetilde{R}_{\overline{1}}^{P Q_1} \cdots \widetilde{R}_{\overline{1}}^{P Q_n} 	= 	\widetilde{R}_{\overline{1}}^{P Q_n} \cdots \widetilde{R}_{\overline{1}}^{P Q_1} \widetilde{R}_{\overline{1}}^{P}.
$$
After conjugating both sides by $Q_n^{-1} P$ we get
$$
	\widetilde{R}_{\overline{1}}^{P Q_n^{-1} P} \widetilde{R}_{\overline{1}}^{P Q_1Q_n^{-1} P} \cdots \widetilde{R}_{\overline{1}}^{P Q_{n-1} Q_n^{-1} P} \widetilde{R}_{\overline{1}} 	= 	\widetilde{R}_{\overline{1}} \widetilde{R}_{\overline{1}}^{P Q_{n-1} Q_n^{-1} P} \cdots \widetilde{R}_{\overline{1}}^{P Q_1Q_n^{-1} P} \	\widetilde{R}_{\overline{1}}^{P Q_n^{-1} P}.
$$
Using the equalities
$$
P Q_n^{-1} P = Q_n,~~P Q_1 Q_n^{-1} P = Q_{n-1}, \ldots, P Q_{n-1} Q_n^{-1} P = Q_1,
$$
we have
$$
	\widetilde{R}_{\overline{1}}^{Q_n} \widetilde{R}_{\overline{1}}^{Q_{n-1}} \cdots \widetilde{R}_{\overline{1}}^{Q_1} \widetilde{R}_{\overline{1}} 	= 	
\widetilde{R}_{\overline{1}} \widetilde{R}_{\overline{1}}^{Q_1} \cdots \widetilde{R}_{\overline{1}}^{Q_{n-1}} \	\widetilde{R}_{\overline{1}}^{Q_n}.
$$
Hence, $	\widetilde{R}$ gives a solution of the \SE.
\end{proof}

\begin{question}
We know that if $R$ is a solution of the YBE, then $S = P R$ is a solution of the braid equation $S_{12} S_{23} S_{12} = S_{23} S_{12} S_{23}$.
What is an analogous of the braid equation for the  {\SE} when $n \geq 3$?
\end{question}

\subsection{Composition of solutions}
For the linear solutions of the YBE we know two operations on solutions (see, for example, \cite{KS}): tensor product and direct sum. If $(V_1, R^{(1)})$ and $(V_2, R^{(2)})$ are two solutions. Then their tensor product is the solution $(V_1 \otimes V_2, R^{(1)} \otimes R^{(2)})$ and the direct sum of solutions is the solution
$(V_1 \times V_2, R^{(1)} + R^{(2)})$, where the map $R^{(1)} + R^{(2)}$ is defined on the basis of $(V_1 \times V_2) \otimes (V_1 \times V_2)$ by the rules
\begin{eqnarray*}
(R^{(1)} + R^{(2)}) (e_i \otimes e_j) &=& R^{(1)}(e_i \otimes e_j),\\
(R^{(1)} + R^{(2)}) (e_i \otimes f_q) &=& e_i \otimes f_q,\\
(R^{(1)} + R^{(2)}) (f_p \otimes e_j) &=& f_p \otimes e_j,\\
(R^{(1)} + R^{(2)}) (f_p\otimes f_q) &=& R^{(2)}(f_p\otimes f_q),
\end{eqnarray*}
where $\{e_{\alpha} \}$ is a basis of $V_1$ and $\{f_{\beta} \}$ is a basis of $V_2$.

 In this subsection we are considering the following question: Let $(X, A)$ and $(Y, B)$ be set-theoretic solutions of the {\SE} and the {\SE[m]}, correspondingly. What new solutions can one construct? 

 If $m=n$, then one can define the direct product of solutions,
$$
A \times B : ( X \times Y)^n \to ( X \times Y)^n,
$$
$$
A \times B \left( (x_1, y_1), (x_2, y_2), \ldots, (x_n, y_n) \right) = \left( (f_1(\bar{x}), g_1(\bar{y})), (f_2(\bar{x}), g_2(\bar{y})), \ldots, (f_n(\bar{x}), g_n(\bar{y})) \right),
$$
where
$$
A(\bar{x}) = (f_1(\bar{x}), f_2(\bar{x}), \ldots, f_n(\bar{x})),~~~B(\bar{y}) = (g_1(\bar{y}), g_2(\bar{y}), \ldots, g_n(\bar{y})),
$$
and we denote
$$
\bar{x} = (x_1, x_2, \ldots, x_n),~~
\bar{y} = (y_1, y_2, \ldots, y_n).
$$
The map $A \times B$ gives a solution of the {\SE}. 

\medskip

\begin{definition} \label{com}
Let $ A \colon  X^{n+k} \to X^{n+k}$ and $ B \colon X^{k+m} \to X^{k+m} $ be two maps which have the  form
$$
	A(\bar{x}, \bar{y}) =
	\left(
	f_1(\bar{x}, \bar{y}), \dotss , f_n(\bar{x}, \bar{y}), h_1(\bar{y}), \dotss , h_k(\bar{y})
	\right),~~~\bar{x} \in X^n,~~\bar{y} \in X^k,
$$

$$
	B(\bar{y}, \bar{z}) =
	\left(
	h_1(\bar{y}), \dotss , h_k(\bar{y}), g_1(\bar{y}, \bar{z}), \dotss , g_m(\bar{y}, \bar{z})
	\right),~~\bar{z} \in X^m,
$$
then a function $A \;\#_k B \colon X^{n+k+m} \to X^{n+k+m}$, defined by
$$
	A \;\#_k B(\bar{x},\bar{y},\bar{z}) :=
	\left(
	f_1(\bar{x}, \bar{y}), \dotss , f_n(\bar{x}, \bar{y}),
	h_1(\bar{y}), \dotss , h_k(\bar{y}),
	g_1(\bar{y}, \bar{z}), \dotss , g_m(\bar{y}, \bar{z})		\right)
$$
is called a $k$-amalgam of $A$ and $B$.
\end{definition}

\begin{theorem}\label{const}
Let $n > 0$, $m>0$, $k \geq 0$, be integers, $A \colon X^{n+k} \to X^{n+k}$ and $ B \colon X^{k+m} \to X^{k+m} $ be solutions of the {\SE[(n+k)]} and the {\SE[(m+k)]}, correspondingly,
$$
	A(\bar{x}, \bar{y}) =
	\left(
	f_1(\bar{x}, \bar{y}), \dotss , f_n(\bar{x}, \bar{y}), h_1(\bar{y}), \dotss , h_k(\bar{y})
	\right),~~~\bar{x} \in X^n,~~\bar{y} \in X^k,
$$
$$
	B(\bar{y}, \bar{z}) =
	\left(
	h_1(\bar{y}), \dotss , h_k(\bar{y}), g_1(\bar{y}, \bar{z}), \dotss , g_m(\bar{y}, \bar{z})
	\right),~~\bar{z} \in X^m,
$$
Then the $k$-amalgam $A \;\#_k B : X^{n+k+m} \to X^{n+k+m}$ of $A$ and $B$,
$$
	A \;\#_k B(\bar{x},\bar{y},\bar{z}) =
	\left(
	f_1(\bar{x}, \bar{y}), \dotss , f_n(\bar{x}, \bar{y}),
	h_1(\bar{y}), \dotss , h_k(\bar{y}),
	g_1(\bar{y}, \bar{z}), \dotss , g_m(\bar{y}, \bar{z})
	\right)
$$ is a solution of the \SE[(n+k+m)] if and only if  for any pair of indices $1 \leqslant i \leqslant n$ and $1 \leqslant j \leqslant m$ and for any collection of elements $a_{\a,\b} \in X$, $\alpha \in \{1, \dots , n+k\}$, $\beta \in \{1, \dots , k+m\}$ the following equality holds:
$$
\begin{gathered}
		f_i\left(\begin{matrix}
g_j(	& a_{1,1},		& a_{1,2},		& \cdots	& a_{1, k-1},	& a_{1, k},		& \cdots	& a_{1, k+m}	& ),\\
g_j(	& a_{2,1},		& a_{2,2},		& \cdots	& a_{2, k-1},	& a_{2, k},		& \cdots	& a_{2, k+m}	& ),\\
	& \vdots		& \vdots		& \ddots	& \vdots		& \vdots 		& \ddots 	& \vdots		& \\
g_j(	& a_{n+1,1},	& a_{n+1,2},	& \cdots	& a_{n+1, k-1},	& a_{n+1, k},	& \cdots	& a_{n+1, k+m}	& ),\\
g_j(	& b^{n+1,1},	& a_{n+2,2},	& \cdots	& a_{n+2, k-1},	& a_{n+2, k},	& \cdots	& a_{n+2, k+m}	& ),\\
g_j(	& b^{n+1,2},	& b^{n+2,2},	& \cdots	& a_{n+3, k-1},	& a_{n+3, k},	& \cdots	& a_{n+3, k+m}	& ),\\
	& \vdots		& \vdots		& \ddots	& \vdots		& \vdots 		& \ddots 	& \vdots		& \\
g_j(	& b^{n+1,k-1},	& b^{n+2,k-1},	& \cdots	& b^{n+k-1,k-1}	& a_{n+k, k},	& \cdots	& a_{n+k, k+m}	& )
		\end{matrix}\right)
		\\
		\veq
		\\
		g_j\left(\begin{matrix}
f_i(	& a_{1,1},		& a_{2,1},		& \cdots	& a_{n+1, 1},	& b_{2, 2},		& b_{3, 2},		& \cdots	& b_{k, 2}		& ),\\
f_i(	& a_{1,2},		& a_{2,2},		& \cdots	& a_{n+1, 2},	& a_{n+2, 2},	& b_{3, 3},		& \cdots	& b_{k, 3}		& ),\\
	& \vdots		& \vdots		& \ddots	& \vdots		& \vdots 		& \vdots 		& \ddots 	& \vdots		& \\
f_i(	& a_{1,k-1},		& a_{2,k-1},		& \cdots	& a_{n+1, k-1},	& a_{n+2, k-1},	& a_{n+3, k-1},	& \cdots	& b_{k, k}		& ),\\
f_i(	& a_{1,k},		& a_{2,k},		& \cdots	& a_{n+1, k},	& a_{n+2, k},	& a_{n+3, k},	& \cdots	& a_{n+k, k}		& ),\\
	& \vdots		& \vdots		& \ddots	& \vdots		& \vdots 		& \vdots 		& \ddots 	& \vdots		& \\
f_i(	& a_{1,k+m},	& a_{2,k+m},	& \cdots	& b^{n+1,k+m}	& a_{n+2, k+m},	& a_{n+3, k+m},	& \cdots	& a_{n+k, k+m}	& )
		\end{matrix}\right),
	\end{gathered}
$$
where
$$
\begin{aligned}
	b^{i, j}	& := h_j(b^{n+1, i-n-1}, \dotss , b^{i-1, i-n-1}, a_{i, j}, \dotss , a_{i, k}),
	\\
	b_{i, j}	& := h_j(a_{n+1, i}, \dotss , a_{n+i, i}, b^{i+1, i}, \dotss , b^{k, i}).
\end{aligned}
$$
\end{theorem}

\begin{proof}
This statement becomes obvious if one looks at the matrix $MI_{n+k+m}$ for the {\SE[(n+k+m)]}. Its $(n+k)\times(n+k+1)$ top-left and $(k+m)\times(k+m+1)$ bottom-right submatrices are independent of the rest of the matrix and represent the {\SE[(n+k)]} and the {\SE[(k+m)]} respectively.

Due to functions $h_i$ overlapping parts of those submatrices are independent of the rest of the matrix and thus the {\SE[(n+k+m)]} recedes to {\SE[(n+k)]} and {\SE[(k+m)]} for indices in corresponding submatrices. Remaining bottom-left and top-right submatrices consist of $nm$ different indices and are transposes of each other.

Explicit form of the {\SE[(n+k+m)]} for indices in those matrices gives us the condition above.
\end{proof}

To illustrate the construction of $k$-amalgam we  consider some examples.

\begin{example}
 Let $X$ be a set and $A, B : X^2 \to X^2$ are solutions of the YBE, which have the forms
$$
A(x, y) = (f_1(x, y), f_2(x, y)), ~~~B(z, t) = (g_1(z, t), g_2(z, t)),~~~x, y, z, t \in X.
$$
By Lemma \ref{YBEF} the maps $f_i$  satisfy the identities
\begin{eqnarray*}
f_1\left( f_1(x, y), z  \right) &=& f_1\left( f_1(x, f_2(y, z)),   f_1(y,  z) \right),\\
f_1\left( f_2(x,y),  f_2(f_1(x,y), z) \right) &=& f_2\left( f_1(x, f_2(y, z)),   f_1(y,  z) \right),\\
f_2\left( f_2(x, y),   f_2(f_1(x, y), z) \right) &=& f_2\left( x,  f_2(y, z)  \right),
\end{eqnarray*}
and the maps $g_i$  satisfy the identities
\begin{eqnarray*}
g_1\left( g_1(x, y), z  \right) &=& g_1\left( g_1(x, g_2(y, z)),   g_1(y,  z) \right),\\
g_1\left( g_2(x,y),  g_2(g_1(x,y), z) \right) &=& g_2\left( g_1(x, g_2(y, z)),   g_1(y,  z) \right),\\
g_2\left( g_2(x, y),  g_2(g_1(x, y), z) \right) &=& g_2\left( x,  g_2(y, z)  \right),
\end{eqnarray*}
for all $x, y, z, t \in X$. Set
$$
R(x, y, z, t) = (f_1(x, y), f_2(x, y), g_1(z, t), g_2(z, t)),~~~x, y, z, t \in X.
$$
Let us find the conditions under which this map $R$ satisfies the {\SE[4]},
$$
	R_{1234} R_{1567} R_{2589} R_{368,10} R_{479,10} = R_{479,10} R_{368,10}  R_{2589} R_{1567} R_{1234}.
$$
The straightforward calculations give
$$
R_{479,10} R_{368,10}  R_{2589} R_{1567} R_{1234} (x, y, z, t, p, q, r, s, u, v) = 
$$
$$
=( f_1(f_1(x,y),p),~~f_1( f_2(x,y), f_2(f_1(x,y),p)),~~f_1(g_1(z,t), g_1(q,r)),~~f_1(g_2(z,t), g_2(q, r)), 
$$
$$
f_2( f_2(x, y), f_2(f_1(x,y),p)),~~f_2(g_1(z,t), g_1(q,r)),~~f_2(g_2(z,t), g_2(q,r)),~~g_1(g_1(s,u),v),
$$
$$
g_1(g_2(s, u), g_2(g_1(s,u),v)),~~g_2(g_2(s, u), g_2(g_1(s,u),v))),
$$
and
$$
R_{1234} R_{1567} R_{2589} R_{368,10} R_{479,10} (x, y, z, t, p, q, r, s, u, v) = 
$$
$$
=( f_1(f_1(x, f_2(y,p)), f_1(y,p)),~~f_2( f_1(x, f_2(y, p)), f_1(y,p)), ~~g_1(f_1(z,q), f_1(t,r)),~~g_2(f_1(z,q), f_1(t, r)), 
$$
$$
f_2(x,  f_2(y,p)),~~g_1(f_2(z,q), f_2(t,r)),~~g_2(f_2(z,q), f_2(t,r)),~~g_1(g_1(s,g_2(u,v)), g_1(u,v)),
$$
$$
g_2(g_1(s, g_2(u,v)), g_1(u,v)),~~g_2(s,  g_2(u,v))).
$$
Equating the left-hand side with the right-hand side gives a system of 10 equalities. The first, second,  and fifth equalities mean that $A$ is a solution of the YBE. The last three  equalities mean that $B$ is a solution of the YBE. Hence,$R$ is a solution of the $4$-SE if and only if
\begin{eqnarray*}
f_1(g_1(z,t), g_1(q,r)) &=& g_1(f_1(z,q), f_1(t,r)),\\
f_1(g_2(z,t), g_2(q, r)) &=& g_2(f_1(z,q), f_1(t, r)),\\
f_2(g_1(z,t), g_1(q,r)) &=& g_2(f_2(z,q), f_2(t,r)),\\
f_2(g_2(z,t), g_2(q,r)) &=& g_2(f_2(z,q), f_2(t,r)).
\end{eqnarray*}
for all $z, t, q, r \in X$. We see that these conditions are conditions on $f_i$ and $g_j$ from Theorem \ref{const}.

\end{example}

\begin{example}
 Let $X$ be a set and $A, B : X^2 \to X^2$ are solutions of the YBE, which have the forms
$$
A(x, y) = (f(x, y), h(y)), ~~~B(y, z) = (h(y), g(y,z)),~~~x, y, z \in X.
$$
By Lemma \ref{YBEF} the maps $f, g$ and $h$ satisfy the identities
$$
f\left( f(x, y), z \right) = f\left( f(x, h(z)), f(y, z) \right),~~~f(h(y), h(z)) = h(f(y, z)),
$$
$$
g\left(y, g(z, t) \right) = g\left(g(y, z), g(h(y), t) \right),~~~g(h(y), h(z)) = h(g(y, z)),
$$
for all $x, y, z, t \in X$. Set
$$
R(x, y, z) = (f(x, y), h(y), g(y, z)) 
$$
is a map from $X^3$ to $X^3$. Then
$$
	R_{123} R_{145} R_{246} R_{356}(x, y, z, t, p, q) = 
$$
$$
= \left( f(f(x, h(t)), f(y, t)),  h(f(y, t)),  g(f(y,t), f(z, p)), h^2(t), g(h(t), h(p)), g(t, g(p, q))\right),
$$
and
$$
R_{356} R_{246} R_{145} R_{123}(x, y, z, t, p, q) = 
$$
$$
= \left( f(f(x, y), t), f(h(y), h(t)),  f(g(y, z), g(t, p)), h^2(t), h (g(t, p)), g(g(t, p), g(h(t), q))\right).
$$
Hence the map $R$ is a solution of the TE if and only if for all $x, y, z, t, p, q \in X$ the following equalities are true
\begin{eqnarray*}
 f\left( f(x, h(t)),   f(y, t)\right)  &=&   f\left( f(x, y), t  \right), \\
 h (f(y, t)) &=& f(h(y), h(t)), \\
g\left( f(y, t),  f(z, p) \right) &=& f \left(  g(y, z),   g(t,  p) \right),\\
g(h(t), h(p)) &=& h (g(t, p)), \\
g\left( t,   g(p, q) \right) &=& g\left( g(t, p),   g(h(t), q)  \right).
\end{eqnarray*}

Since $A$ and $B$ are solution of the YBE the first, second, forth and fifth equalities are satisfied. Hence, $R$ is a solution of the TE if and only if for all $y, z, t, p \in X$ holds
$$
g\left( f(y, t),  f(z, p) \right) = f \left(  g(y, z),   g(t,  p) \right).
$$

On the other side, in this situation $n = m = k = 1$ and by Theorem \ref{const} the map $A \;\#_1 B : X^3 \to X^3$,
$$
A \;\#_1 B (x, y, z) = (f(x, y), h(y), g(y, z)) 
$$
gives a solution for the TE if and only if the following equality true
$$
g\left( f(a_{11}, a_{21}),  f(a_{12}, a_{22}) \right) = f \left(  g(a_{11}, a_{12}),  g(a_{21}, a_{22}) \right).
$$
\end{example}

\begin{remark}
Note that for $k \in \{0, 1\}$ the conditions of Theorem \ref{const} does not contain elements $b^{i, j}$ and $b_{i, j}$. Moreover for $k=0$ $ A \;\#_0 B = A \times B $. It is easy to see that in this case if the function $A \times B$ is a solution of the {\SE} then the function $B \times A$ is a solution as well.
\end{remark}

\medskip

\subsection{Simple solutions} In this subsection we define some class of  solutions to $n$-simplex equations.

\begin{definition}
A solution $T : X^n \to X^n$ of the  {\SE} is said to be simple if
$$
	T(x_1, \dotss , x_n) = (x_{s(1)}, \dotss , x_{s(n)}),
	\text{where } s: \{1, \dotss , n\} \to \{1, \dotss , n\} \text{ is a map (not necessary injective)}.
$$  
\end{definition}

\begin{example} \label{simple}
Let $X$ be a set.

1) The identity map $\id : X \to X$ is a simple   solution to the $1$-SE.

2) The maps $P, Pr^2_1, Pr^2_2 : X^2 \to X^2$ that are defined by the rules
	$$
	\begin{aligned}
		P : & & (x, y) & \mapsto (y, x),
		\\
		Pr^2_1 : & & (x, y) & \mapsto (x, x),
		\\
		Pr^2_2 : & & (x, y) & \mapsto (y, y),
	\end{aligned}
	$$
for all $x, y, z \in X$, are simple solutions to the $2$-SE.

3)	The map $Pr^3_2 : X^3 \to X^3$,
		$Pr^3_2 :  (x, y, z)  \mapsto (y, y, y)$,$x, y, z \in X$, is a simple solutions to the $3$-SE.
\end{example}

It is evident that the direct product of simple solutions is a simple solution.

From Theorem \ref{const} follows

\begin{proposition}
	Let $ R : X^n \to X^n $ be a solution of the {\SE}
	and $A$ be any map from the list $\{ \id_X, P, Pr^2_1, Pr^2_2, Pr^3_2 \} $.
	Then $ R \times A $ and $ A \times R$ are solutions of the {\SE[(n+m)]}.
\end{proposition}

It seems that any simple solution can be  construct from  the list $\{ \id_X, P, Pr^2_1, Pr^2_2, Pr^3_2 \} $ by a sequence of direct products.

\begin{definition} A solution $R$ is called indecomposable if there are no solutions $A$ and $B$, such that $R = A \times B$.
\end{definition}

\begin{question}
1)	Is there a indecomposable simple solution of the {\SE} for some $n$ different from five solutions $\{ \id_X, P, Pr^2_1, Pr^2_2, Pr^3_2 \} $?

2) We know that the permutation $P_{12}$ is a solution of the YBE. For which  $n > 2$ there are non-identity permutations without fixed points that gives solutions of the \SE? 

Of course, using Proposition \ref{invsymm} it is not difficult to find some transpositions which are solutions of \SE. 
\end{question}

\subsection{Linear and affine solutions} 

 Linear and affine solutions of the YBE and the TE are studied in \cite{HLinear, Bukh, KNPT, BIKP}.

Using Theorem  \ref{const} we are proving

\begin{proposition} \label{sumlin}
Let $A$ and $B$ be linear solutions of the \SE  \,  and the \SE[m], respectively. Then their $0$-amalgam $A\times B$ is a linear solution of the \SE[(n+m)]. If $A$ and $B$ such that  the $1$-amalgam   $A \#_1 B$ is defined, then it is a solution of the \SE[(n+m-1)].
\end{proposition}

\begin{proof} Let $f \colon X^n \to X$ be a linear map. Then it can be presented in terms of matrix multiplication
$$
	f(\bar{x}) = a_1 x_1 + a_2 x_2 + \ldots + a_n x_n = \bar{x}^T [f] = [f]^T\bar{x},
$$
	where $[f]^T = (a_1, a_2, \ldots,  a_n)$ is a vector of coefficients, $\bar{x}^T = (x_1, x_2, \ldots ,  x_n)$ is a vector of variables and $\cdot^T$ is the transposition. It is easy to see that the condition of Theorem  \ref{const} for $k \in \{0, 1\}$ can be rewritten as
	$$
		[f]^T M [g] = [g]^T M^T [f],
	$$
	where $ f $ is a component of the map $A$, $g$ is a component of the map $B$ and $ M $ ia an $n \times m$ matrix and so the condition holds.

	This means that for any linear solutions $A$ and $B$ of some \SE s the map $A \;\#_k B$ for $k \in \{0, 1\}$ will  be a solution of some \SE. 
\end{proof}

\begin{question} 
1) What can we say on $k$-amalgam of affine solutions? Is it true that for $k=0$ and/or $k=1$ we get a solution? 

2) Under which conditions $k$-amalgam of rational solutions gives a rational solution?
\end{question}

\subsection{Constructing rational solutions from the linear solutions}

Take a linear solution $(\mathbb{R}, R)$ of the  YBE:
$$ R(x, y) = (\alpha_1 x, (1-\alpha_1 \alpha_2)x + \alpha_2 y), \text{ where } \alpha_1, \alpha_2 \in \mathbb{R}. $$ 
By proposition~\ref{sumlin} we have a solution $R : \mathbb{R}^n \to \mathbb{R}^n $ of an \SE{}, defined by
	$$ 
	x_i \mapsto \begin{cases}
	\alpha_i x_i & \text{ if } i \text{ is odd,}\\
	(1 - \alpha_{i-1}\alpha_i)x_{i-1} + \alpha_i x_i + (1 - \alpha_i \alpha_{i+1})x_{i+1} & \text{ if } i \text{ is even,}\\
(1 - \alpha_{i-1}\alpha_i)x_{n-1} + \alpha_i x_n & \text{ if } i=n \text{ is even,}\\
	\end{cases} 
	$$
	where $\alpha_i \in \mathbb R$. 

	Take any rational function $f: \  \mathbb{R} \to \mathbb{R}$ having a rational inverse, for example a linear fractional transformation
		$$ f(x) = \dfrac{ax + b}{cx + d}, \quad f^{-1}(x) = \dfrac{dx - b}{a - cx}.$$
	We can conjugate our solution by $f$. This yields a rational solution $(f^{-1})^n R f^{n}$. In the case of a linear rational transformation, this conjugation takes the form
		$$ x_i \mapsto \begin{cases} \dfrac{(\alpha_i ad - bc)x_i + (\alpha_i - 1)bd}{ca(1-\alpha_i) x_i + (ad - \alpha_i bc)}& \text{ if } i \text{ is odd,}\\
		& \\
		\dfrac{d \left(\beta_i \tfrac{a x_{i-1} + b}{c x_{i-1} + d}+ \alpha_{i} \tfrac{a x_{i} + b}{c x_{i} + d} + \beta_{i+1}\frac{a x_{i+1} + b}{c x_{i+1} + d}\right) - b}{a - c \left(\beta_i\tfrac{a x_{i-1} + b}{c x_{i-1} + d}+ \alpha_{i} \tfrac{a x_{i} + b}{c x_{i} + d} + \beta_{i+1} \tfrac{a x_{i+1} + b}{c x_{i+1} + d}\right)} & \text{ if } i \text{ is even, }\\
& \\
		\dfrac{d \left(\beta_n \tfrac{a x_{n-1} + b}{c x_{n-1} + d}+ \alpha_{n} \tfrac{a x_{n} + b}{c x_{n} + d} \right) - b}{a - c \left( \beta_n\tfrac{a x_{n-1} + b}{c x_{n-1} + d}+ \alpha_{n} \tfrac{a x_{n} + b}{c x_{n} + d} \right) } & \text{ if } i=n \text{ is even, }		
		\end{cases} $$
		where $\beta_i = 1 - \alpha_{i-1}\alpha_i$ for $i \in \{2,3, \ldots , n\}$.
	This construction provides lots of examples of new rational solutions of arbitrary {\SE} and one can wonder whether or not some solution is a conjugate to a linear one.
\begin{example} If we put
	$$ \alpha_1 = \alpha_3 = 1, \, \alpha_2 = \alpha_4 = 0 $$
	then $$\beta_2 = \beta_3 = \beta_4 = 1.$$
Also let $$a = c = 1, b = 0, d = -1,$$ then we get a rational solution of the 4-SE
$$(x_1, x_2, x_3, x_4) \mapsto \left(x_1, \dfrac{x_1 + x_3 - 2 x_1 x_3}{1 - x_1 x_3}, x_3, x_3\right).$$
\end{example}
\begin{example} Let $R$ be a linear solution of the {\SE[4]} defined by
	$$  R (x_1, x_2, x_3, x_4) = (x_2 - x_4, x_1 + x_3, x_3, x_4).$$
By taking $f$ equal to $ \dfrac{x}{x-1}$ we get a rational solution of the {\SE[4]}
	$$ (f^{-1})^n R f^n (x_1, x_2, x_3, x_4) = \left(\dfrac{x_2 - x_4}{1 + x_2 x_4}, \dfrac{x_1 + x_3 - 2x_1x_3}{1 - x_1 x_3}, x_3, x_4\right). $$

\end{example}

	\question{Can we obtain the electric solution by conjugation of a linear one?}


\subsection{From solutions of the {\SE} to solutions of the {\SE[(n-1)]}} We defined a $k$-amalgam of solutions, which gives a possibility to construct solutions of higher dimensions. Now we are going to the  opposite direction: If $ (X, R) $ is a solution of the {\SE} is it possible to construct a solution of the {\SE[(n-1)]}?

\begin{proposition}
	Let $ R : X^n \to X^n $, $n \geq 3$, be a solution of the {\SE},
	and there exists $x_0 \in X$ such that $ R(x_0, \dotss , x_0) = (x_0, \dotss , x_0)$.
	Then
	$$
		R^r(x_1, \dotss , x_{n-1}) := R(x_0, x_1, \dotss , x_{n-1})
	$$
	and
	$$
		R^l(x_1, \dotss , x_{n-1}) := R(x_1, \dotss , x_{n-1}, x_0)
	$$
	are solutions of the {\SE[(n-1)]}.
\end{proposition}
\begin{proof}
	If $R$ is a solution of the {\SE},  then the maps ${R}_{\bar{i}}$, $1 \leq i \leq n+1$, act on $X^N$, where $N = n (n+1)/2$.
	Let us consider a set $X_{x_0} \subset X^N$ with first $n$ coordinates equal to $x_0$,
	$$ X_{x_0} = \{ (\underbrace{x_0, \ldots , x_0}_n, x_{n+1}, x_{n+2}, \ldots, x_N) \in X^N \}. $$
	We restrict the {\SE}  to the set $X_{x_0}$.
	In the right hand side  of this equation we apply at first ${R}_{\bar{1}}$  
	$$
	{R}_{\bar{1}} (\underbrace{x_0, \ldots , x_0}_n, x_{n+1}, x_{n+2}, \ldots, x_N) = R(x_0, \ldots, x_0) \times \id^{N-n}( x_{n+1}, x_{n+2}, \ldots, x_N) = 
$$
$$		
=(\underbrace{x_0, \ldots , x_0}_n, x_{n+1}, x_{n+2}, \ldots, x_N). 
		$$
Hence, ${R}_{\bar{1}} $	acts identically  on $ X_{x_0}$ and on this set $ X_{x_0}$ the {\SE} has the form
		$${R}_{\bar{1}} {R}_{\bar{2}} \ldots {R}_{\overline{n+1}} = {R}_{\overline{n+1}} \ldots {R}_{\bar{2}}.$$
Both sides  of the {\SE[(n-1)]} act on $X^{N'}$ where $N' = N - n.$
If we remove the first $n$ coordinates and   consider the equation above restricted to the last $N'$ coordinate, we get a solution of the {\SE[(n-1)]} after shifting all indices by $-n$. 

To prove that 	$R^l$ is a 	solution of   the {\SE[(n-1)]} it is need to take the subset of $X^N$ in which the last $n$ coordinates are equal to $x_0$ and remark that the map $R_{\overline{n+1}}$ acts identically on this subset.

\end{proof}

\bigskip

\section{Tropicalization} \label{Trop}

\subsection{Tropicalization of rational solutions} By tropicalization we mean a map from the set of rational function to the set of piecewise linear functions.
I.~Dynnikov \cite{D} used tropicalization to study representations of the braid group $B_n$ into the permutation group $Sym(\mathbb{Z}^{2n})$. In particular,  he finds a non-degenerate solution of the YBE over $\mathbb{Z}^{2}$.

In the present section we find a connection between rational solutions of {\SE}s and piecewise linear solutions of {\SE}s.
To formulate the main result, we introduce some definitions and notations.
Let $\mathbb{R}(x_1, x_2, \ldots, x_n)$ be  the field of rational fractions with coefficients in the field of real numbers $\mathbb{R}$. Any solution $(\mathbb{R}, R)$ of  the {\SE}, where
$$
R(x_1, x_2, \ldots, x_n) = (r_1, r_2, \ldots, r_n),~~~r_i \in \mathbb{R}(x_1, x_2, \ldots, x_n),
$$
 is said to be a {\it rational solution}.

Let us denote by $I_n$ a subset of non-zero fractions $r = f/g \in \mathbb{R}(x_1, x_2, \ldots, x_n)$ such that all coefficients of the numerator $f$ are equal to $1$ and the free term  is equal to zero, the denominator  $g$ is equal to $1$ or all its coefficients are equal to $1$ and the free term is equal to zero. 
A rational solution
$$
R(x_1, x_2, \ldots, x_n) = (r_1, r_2, \ldots, r_n),~~~r_i \in \mathbb{R}(x_1, x_2, \ldots, x_n),
$$
of the {\SE} is said to be an $I${\it--rational solution} if all components $r_i$ lie in $I_n$.

\begin{example} \label{ex1}
It is easy to see that the so called {\it electric solution} of the TE,
\begin{align*}
	R_E(x,y,z) &= \left(\frac{xy}{x+z+xyz},\ x+z+xyz,\ \frac{yz}{x+z+xyz}\right)\\
\end{align*}
that appears in the theory of electric circuits
as the famous Y-$\Delta$ transform of an electric network is an $I$--rational solution.

The map
\begin{align*}
	R_e(x,y,z) &= \left(\frac{xy}{x+z},\ x+z,\ \frac{yz}{x+z}\right)
\end{align*}
which is obtained from $R_E$ by removing terms of degree three is also an $I$--rational solution.
\end{example}

Let us  denote by $PL_n$ the set of piecewise linear functions  $\mathbb{R}^n \to \mathbb{R}$.

\begin{definition} The {\it tropicalization} is a map $(\cdot)^t : I_n \to PL_n$ that is defined on $r = f/g \in I_n$, where 
$$
f = \sum_{i_1 + \ldots + i_n > 0} \alpha_{i_1\ldots i_n} x_1^{i_1} \ldots x_n^{i_n},~~~
g = \sum_{j_1 + \ldots +j_n \geq 0} \beta_{j_1\ldots j_n} x_1^{j_1} \ldots x_n^{j_n}
$$
by the rule
$$
r^t = \begin{cases}
\underset{i_1+\ldots + i_n > 0}{\max} \{ i_1 x_1 + \ldots + i_n x_n \} - \underset{j_1+\ldots +j_n > 0}{\max} \{ j_1 x_1 + \ldots +j_n x_n \}, & \text{for $g \not= 1$;} \\
\underset{i_1+\ldots + i_n > 0}{\max} \{ i_1 x_1 + \ldots + i_n x_n \}, & \text{for $g= 1$.}
\end{cases}
$$

\end{definition}

Note that $r^t$ can be  obtained from $r$ using the following recursive procedure.

\begin{proposition} \label{rec}
	Let $r=r(x_1,\ldots,x_n)$, $r_1=r_1(x_1,\ldots,x_n)$, $r_2=r_2(x_1,\ldots,x_n)$ be  rational functions from $I_n$.
	Then
\begin{enumerate}
	\item
	if $r= x_i$, then $r^t = x_i$, for $i=1,\ldots,n$;
	\item
	$(r_1 + r_2)^t = \max\{r^t_1,\ r^t_2\}$;
	\item
	$(r_1 r_2)^t = r_1^t + r_2^t$;
	\item
	$\left(\dfrac{r_1}{r_2}\right)^t = r_1^t - r_2^t$.
\end{enumerate}
In particular,	 tropicalization is correctly defined i.e.	if  $r_1$ and $r_2$ are two equal rational functions,  then their tropicalizations are also equal functions.
\end{proposition}

Let $R(x_1,\ldots,x_n) = (r_1, r_2, \ldots, r_n) \in (I_n)^n$ be a rational vector-valued map of $n$ variables.
Define the {\it tropicalization} of the rational map $R$ componentwise:
$$
	R^t(x_1,\ldots,x_n) := (r_1^t(x_1,\ldots,x_n),\ldots,r_n^t(x_1,\ldots,x_n)).
$$

\begin{example}
The tropicalization of the electric solution $R_E$ from Example \ref{ex1} is
$$
R_E^t(x, y, z) = (x + y - \max\{x,z,x+y+z\},\ \max\{x,z,x+y+z\},\ y + z - \max\{x,z,x+y+z\}).
$$
The tropicalization of  $R_e$ from Example \ref{ex1} is
$$
R_e^t(x, y, z) = (x + y - \max\{x,z\},\ \max\{x,z\},\ y + z - \max\{x,z\}).
$$
It is easy to see that $R_E^t$ and $R_e^t$ are piecewise  linear solutions of the  TE.
\end{example}

\begin{remark}
The solution $R_E^t$ consists of 3 linear pieces:
$$
R_1(x, y, z) = (y, x, y + z - x),
$$
$$
R_2(x, y, z) = (x + y - z,  z,  y),
$$
$$
R_3(x, y, z) = (- z, x+y+z, -x).
$$
We can consider these maps as maps on $\mathbb{R}^3$. One can check that $(\mathbb{R}, R_1)$ and $(\mathbb{R}, R_2)$ are solutions of the TE, but $(\mathbb{R}, R_3)$
is not.
\end{remark}

\begin{definition} Let $r_1$ and $r_2$ be two rational functions in $\mathbb{R}(x_1, \ldots, x_n)$ and 
 $1 \leqslant k \leqslant n$. A $k$-\textit{composition} of $r_1$ and $r_2$ is a rational function $ r_1 \circ_k r_2 \in \mathbb{R}(x_1, \ldots, x_n)$ defined as
$$
		(r_1 \circ_k r_2)(x_1,\ldots,x_n) := r_1(x_1,\ldots,x_{k-1},r_2(x_1,\ldots,x_n),x_{k+1},\ldots,x_n).
$$
\end{definition}

\begin{proposition} If $r_1, r_2 \in I_n$, then
	$(r_1 \circ_k r_2)^t = r_1^t \circ_k r_2^t$.
\end{proposition}

\begin{proof}
	The proof uses the recursive procedure given in Proposition \ref{rec}.
	If $r_1(x_1,\ldots,x_n) = x_i$ is a variable, then the statement is obvious.
	If 
	$$
	r_1(x_1,\ldots,x_n) = f(x_1,\ldots,x_n) * g(x_1,\ldots,x_n),
	$$
	 where $f$ and $g$ are rational functions in $I_n$
	and $*~\in~\{+,\cdot,/\}$ is a binary operation in rational function $r_1$, then the tropicalization
	$$
		r^t_1(x_1,\ldots,x_n) = f^t(x_1,\ldots,x_n) *^t g^t(x_1,\ldots,x_n),
	$$ where $*^t~\in~\{\text{max},+,-\}$, and
	$$
		(r_1 \circ_k r_2)(x_1,\ldots,x_n) = (f \circ_k r_2)(x_1,\ldots,x_n) * (g \circ_k r_2)(x_1,\ldots,x_n).
	$$
	Hence, the tropicalization
	$$
		(r_1 \circ_k r_2)^t(x_1,\ldots,x_n) = (f \circ_k r_2)^t(x_1,\ldots,x_n) *^t (g \circ_k r_2)^t(x_1,\ldots,x_n).
	$$
The needed statement could be obtained by induction on a  recursion depth.
\end{proof}

\begin{corollary} \label{cortr}
	Tropicalization preserves a composition of rational vector-valued maps of the form described above.
\end{corollary}

\begin{theorem}
	If $(\mathbb{R}_{>0}, R)$, $R \in (I_n)^n$ is a solution of the $n$-simplex equation on the set of positive real numbers, 
	then its tropicalization $(\mathbb{R}_{>0}, R^t)$ is also a solution of the $n$-simplex equation.
\end{theorem}

\begin{proof}
	For any $n \geqslant 2$, the $n$-SE has the following form
	$$
	R_{\overline{1}}R_{\overline{2}} \cdots R_{\overline{n+1}}
	=
	R_{\overline{n+1}} \cdots R_{\overline{2}} R_{\overline{1}}.
	$$
	Apply the tropicalization procedure to both sides:
	$$
		(R_{\overline{1}}R_{\overline{2}} \cdots R_{\overline{n+1}})^t = (R_{\overline{n+1}} \cdots R_{\overline{2}} R_{\overline{1}})^t.
	$$
	Hence, according to Corollary \ref{cortr}, the following relation holds:
	$$
	R_{\overline{1}}^t R_{\overline{2}}^t \cdots R_{\overline{n+1}}^t
	=
	R_{\overline{n+1}}^t \cdots R_{\overline{2}}^t R_{\overline{1}}^t.
	$$
	It means that $R^t$ is a solution of the $n$-simplex equation.
\end{proof}

\subsection{Formal tropicalization of solutions on algebraic systems}

The reasonable question that arises is whether we can make a replacement, analogous to a tropicalization, in solutions of the {\SE} such that we again obtain a solution.
Obviously this question has a positive answer.
Such substitutions of functions in formulas of solutions that transform solutions into solutions we will call formal tropicalizations.
In order to give more formal definition we shall recall the notion of a partial algebra.
A partial algebra is a generalization of universal algebra that allows partial operations \cite[ch. 2]{Gratz}.
A natural example of partial algebra is a field,  a multiplicative inverse is not defined for element $0$. 

\begin{definition}
	Let $ \mathcal{A} $ be a partial algebra of signature $ \Sigma $ and $ R : A^n \to A^n $ be a partial function that can be
	expressed as a realization of a collection $ (f_1, \dotss, f_n) $ of $n$-ary functions--terms of a signature
	$ \Sigma' \subset \Sigma $ such that $ (A, R) $ is a solution of the \SE.
	Let $ \mathcal{B} $ be a partial algebra of a signature $ \Xi $,
	and $ g $ is a mapping that match every $n$-ary functional symbol from $ \Sigma' $ with $n$-ary function of the signature $\Xi$.
	Then for every $ f_i $ we can define $ g(f_i) $ as $n$-ary function of the  signature $\Xi$ where every functional symbol from
	$ \Sigma' $ is replaced with corresponding $n$-ary function of the  signature $\Xi$.
	The function $ R^g : B^n \to B^n $ that is realization of a collection $ (g(f_1), \dotss, g(f_n)) $ will be called a
	{\it formal tropicalization } of the solution $ R $ if $ (B, R^g) $ is a solution of the \SE.
\end{definition}

The sentence ``$(A, R) $ is a solution of the \SE'' could be understood as a system of $ N $ equations $\Phi_R$ with $ N $ variables such that for any interpretation of variables every equation is either true or at least one side is undefined.
This is just an explicit form of the {\SE}.

Those equations can be viewed as formulas of equational logic.
For some survey of equational logic in the context of algebras one can see \cite[appendix 4]{Gratz}.

Following the notation above, let $T$ be a set of equational formulas of signature $\Sigma'$ such that
$\mathcal{A}$ satisfies $T$ and every formula $\Phi_R$ has finite proof starting from $T$.
These formulae will be called {\it axioms} of the solutions $(A, R)$.
\begin{proposition}
	Let $ (A , R) $ be a solution of the {\SE} and $ T $ its set of axioms.
	If $ g $ is a mapping such that $ \mathcal{B} $ satisfies $ g(T) := \{ g(\phi) : \phi \in T \} $
	then $ (B, R^g) $ is a formal tropicalization of the solution $ (A, R) $, i.e. it is a solution of the \SE.
\end{proposition}
\begin{proof}
	It is easy to see that the set $ g(\Phi_R) $ coincide with the set $\Phi_{R^g}$.
	Now since $\Phi_R$ is derivable from $T$ the set $g(\Phi_R)$ is derivable from $g(T)$.
	From these two facts one can conclude that $\mathcal{B}$ satisfies $\Phi_{R^g}$.
\end{proof}

\begin{example}
	Let us consider the case of the tropicalization of rational functions described above.
	We can interpret rational solution as a solution in algebraic system $(X, \cdot^2, /^2, +^2, 1^0)$.
	All those rational solutions can be proven from the following set of axioms:
	$$
	\begin{gathered}
		a \cdot b = b \cdot a;
		\qquad
		a + b = b + a;
		\\
		a \cdot (b \cdot c) = (a \cdot b) \cdot c;
		\qquad
		a + (b + c) = (a + b) + c;
		\qquad
		a \cdot (b + c) = a \cdot b + a \cdot c;
		\\
		a / b = a \cdot (1 / b);
		\qquad
		(1 / a) \cdot (1 / b) = 1 / (a \cdot b);
		\qquad
		a \cdot 1 = a;
		\qquad
		a / 1 = a;
		\qquad
		a / a = 1.
	\end{gathered}
	$$
	Thus we can construct mapping $g$ to signature of any algebraic system that meet those axioms.
	One of such algebraic systems is $(X, +^2, -^2, \max^2, 0^0).$
	
	Note that as long as we consider multiplication by constant natural number as iterated operation of addition we can use the same
	mapping $g$ for rational functions with natural coefficients,
	and with some formal accuracy even for rational functions with positive rational coefficients.
\end{example}

\begin{corollary}
	Let $(\mathcal{A}, R)$ be a solution of the {\SE} and $\mathcal{B}$ be some subsystem of an algebraic system $\mathcal{A}$
	in the same class defined by the signature $\Sigma$.
	Then $(\mathcal{B}, R)$ is a solution of the {\SE} as well.
\end{corollary}

\begin{proposition} \label{Tbyhom}
	Let $(\mathcal{A}, R)$ be a solution of the {\SE} and $\mathcal{B}$ be a homomorphic image of an algebraic system
	$\mathcal{A}$ under the homomorphism $h$.
	Then $(\mathcal{B}, R^h)$ is a solution of the {\SE}.
\end{proposition}

One of the advantage of formal tropicalization is that it allows us to see solutions of the {\SE} as some templates rather than fixed functions.
Consider the following example.
\begin{example}
	We know that
	$$
		(x, y) \mapsto (ax, by + (1-ab)x), \text{ where } a,b \in \mathbb{Z}
	$$
	is a solution of the YBE over $ \mathbb{Z} $.
	Since $ \mathbb{Z} $ as a free left $ \mathbb{Z} $-module is generic for the theory of left modules over commutative rings,
	we can see that for any commutative ring $ R $
	$$
		(x, y) \mapsto (ax, by + (1-ab)x), \text{ where } a,b \in R
	$$
	is a solution of the YBE over any left $R$-module.
\end{example} 
From the Proposition \ref{Tbyhom} we obtain the following result.
\begin{corollary} 
	Let $(\mathcal{A}, R)$ be a solution of the {\SE} and $T$ its set of axioms.
	Consider the class $\mathfrak{K}$ of algebraic systems defined by the subset of axioms $ T' \subset T $ that contain variables.
	For any system $\mathcal{B} \in \mathfrak{K}$ if there is homomorphism
	$ h : \mathcal{A} \to \mathcal{B}$ then $ (\mathcal{B}, R^h) $ is a solution of the {\SE}.
\end{corollary}

\begin{example}
	Let $\mathfrak{R}$ be the class of commutative rings with partial operation of multiplicative inverse.
	Consider a rational solution $(\mathbb{R}, R)$ of the {\SE} of the form $(R_1, \dotss , R_n)$ with components
	$$
		R_i =
		\frac{
			\sum a_j \cdot x_1^{\a_{j, 1}} \cdot ... \cdot x_n^{\a_{j, n}}
		}{
			\sum b_k \cdot x_1^{\b_{k, 1}} \cdot ... \cdot x_n^{\b_{k, n}}
		},
		\text{ where }
		\alpha_{j, m}, \beta_{k, m} \in \mathbb{N}\cup\{0\} \text{ for all } j, k, m.
	$$
	In the class $\mathfrak{R}$ there is an obvious homomorphism $h : \mathbb{R} \to \mathcal{D}'[\mathbb{R}^m]$
	from the field $\mathbb{R}$ of real numbers with standard operations of addition and multiplication
	to the linear space $\mathcal{D}'[\mathbb{R}^m]$ of generalized functions with compact support
	on the space $\mathbb{R}^m$ with the standard operation of addition and the convolution $*$ as the multiplication.
	This homomorphism acts by
	$$
		r  \mapsto r \cdot \delta \in \mathcal{D}'[\mathbb{R}^m],~~~r \in \mathbb{R},
	$$
	where $\delta$ is the Dirac delta.
	Thus $(\mathcal{D}'[\mathbb{R}^m], R^h)$ is a rational solution of the {\SE}.
	Components of the function $R^h$ are
	$$
		R^h_i =
		\left(
			\sum a_j \cdot x_1^{\a_{j, 1}} * ... * x_n^{\a_{j, n}}
		\right)\left(
			\sum b_j \cdot x_1^{\b_{j, 1}} * ... * x_n^{\b_{j, n}}
		\right)^{-1},
	$$
	where powers are repeated convolutions and the inverse is taken with respect to the convolution.
\end{example}

\section{Solutions of the parametric YBE} \label{GrExt}

\subsection{Group extensions and the YBE} \label{Gr}

Suppose there exists a group extension
\FloatBarrier
\begin{figure}[!ht]
\centering
\begin{tikzcd}[ampersand replacement=\&]
	1 \ar[r, ""] \& H \ar[r, "i"] \& G \ar[r, "j"] \& K \ar[r, ""] \& 1
\end{tikzcd}
\end{figure}
\FloatBarrier
\noindent and a section $\varphi : K \to G$ that is a map such that $j \circ \varphi = \id_{K}$. 

In \cite{PT} was prove that if $K$ is abelian, then the conjugacy quandle $Conj(G)$ defines a solution of the parametric YBE, 
$$
R_{12}^{a,b} R_{13}^{a,c} R_{23}^{b,c} = R_{23}^{b,c} R_{13}^{a,c} R_{12}^{a,b},~~~a, b, c \in K,
$$
on $H$ with parameters in $K$. 

In this subsection we generalize this result. To do it we recall some facts on group extensions. We will 
present  elements of $G$ as the set of pairs $(x, a) \in H \times K$ and assume that $i(x) = (x,1)$ and $j((x, a)) = a$. Then the multiplication on $G$ is defined by the rule
$$
(x, a) \circ (y, b) = (x \underset{a,b}{\circ} y, a \circ b)~~\mbox{for some}~x \underset{a,b}{\circ} y\in H,
$$
where $a \circ b$ is the multiplication on $K$. For the multiplication on $H$ we use the same symbol  $\circ$. These agree with the formula
$$
(x, 1) \circ (y, 1) = (x \circ y, 1).
$$

Hence, we can consider $H$ as algebraic system with the set of binary algebraic operations
$$
\{ \underset{a,b}{\circ} ~|~ a, b \in K \}.
$$
From group axioms follows that these operations satisfies the following axioms
\begin{enumerate}
\item for any $x \in H$ and $a \in K$ there is unique element $x_{a,a^{-1}}^{-1} = x_{a^{-1},a}^{-1} \in H$ such that
$$
x_{a^{-1},a}^{-1} \underset{a^{-1},a}{\circ} x = x  \underset{a, a^{-1}}{\circ}  x_{a,a^{-1}}^{-1}  = 1;
$$
\item for any $x, y, z \in H$ and $a, b, c \in K$ holds
$$
(x \underset{a,b}{\circ} y) \underset{a\circ b,c}{\circ} z = x \underset{a,b \circ c}{\circ} (y \underset{b,c}{\circ} z).
$$
\end{enumerate}

Now we generalize this construction to the groups with a structure of a right distributive groupoid.
Suppose, that on the set $G$ is defined a binary algebraic operation $* : G \times G \to G$ such that $(G, *)$ is a right distributive groupoid, the set $H$ is closed under multiplication $*$ and this multiplication defines a right distributive groupoid on $K$. As we know, the map
$$
R(g, h) = (g, h*g), ~~g, h \in G,
$$
defines a solution of the YBE on $G$. For example, in \cite{PT} the operation $*$  is $a * b = b^{-1} a b$, $a, b \in G$, it means that $(G,*)$ is a conjugation quandle.

If we present elements of $G$ as the set of pairs $(x, a) \in H \times K$, then 
$$
(x, a) * (y,b) = (x \underset{a,b}{*} y, a * b)~~\mbox{for some}~x \underset{a,b}{*} y \in H.
$$
Hence, on $H$ we have operation $*$ and a set of operations $\{ \underset{a,b}{*}~|~a, b \in K \}$. From the right distributivity of $*$ follows

\begin{lemma} 
For any $x, y, z \in H$ and $a, b, c \in K$ holds
$$
(x \underset{a,b}{*} y) \underset{a* b,c}{*} z = (x \underset{a,c}{*} z)  \underset{a * c, b * c}{*}   (y \underset{b,c}{*} z).
$$
\end{lemma}

Now we are ready to prove

\begin{proposition}
For any $a, b, c \in K$ the following equality
$$
R_{12}^{a,b} \, R_{13}^{a,c*b} \, R_{23}^{b,c} = R_{23}^{b*a,c*a} \, R_{13}^{a,c} \, R_{12}^{a,b} 
$$
holds in $H$, where
$$
R^{u,v}(x, y) = (x, y \underset{v,u}{*} x),~~u, v \in K.
$$
\end{proposition}

\begin{proof}
Since $(G, R)$ is a solution of the YBE, then for any $g, h, k \in G$ holds
$$
R_{12} \, R_{13} \, R_{23} (g, h, k) = R_{23} \, R_{13} \, R_{12} (g, h, k).
$$
Let $g = (x, a)$, $h = (y, b)$,  $k = (z, c)$. Then the left  side,
$$
R_{12} \, R_{13} \, R_{23} ((x, a),  (y, b),  (z, c)) = R_{12} \, R_{13} ((x, a),  (y, b),  (z, c) * (y, b)) =
$$
$$
= R_{12} \, R_{13} \left((x, a),  (y, b),  ( z \underset{c,b}{*} y, c* b )\right)  = R_{12} \left((x, a),  (y, b),  ( (z \underset{c,b}{*} y) \underset{c*b,a}{*} x, (c* b)*a )
\right) =
$$
$$
=\left((x, a),  (y \underset{b,a}{*} x, b* a),  ( (z \underset{c,b}{*} y) \underset{c*b,a}{*} x, (c* b)*a ) \right).
$$
The right side,
$$
R_{23} \, R_{13} \, R_{12} ((x, a),  (y, b),  (z, c)) = R_{23} \, R_{13} \left((x, a),  (y \underset{b,a}{*} x, b * a),  (z, c)\right) =
$$
$$
= R_{23} \left((x, a),  (y \underset{b,a}{*} x, b * a),  (z \underset{c,a}{*} x, c * a)\right) = 
$$
$$
= \left((x, a),  (y \underset{b,a}{*} x, b * a),  ((z \underset{c,a}{*} x) \underset{c*a,b*a}{*} (y \underset{b,a}{*} x), (c * a) * (b * a))\right).
$$

If we restrict the action of $R$ on $H \times H$ and put
$$
R^{u,v}(x, y) = (x, y \underset{v,u}{*} x),~~u, v \in K,
$$
then we get the need equality.
\end{proof}

\begin{corollary} \label{trext}
If $(K, *)$ is a trivial right distributive groupoid, i.e. $u * v = u$ for any $u, v \in K$, then for any $a, b, c \in K$ the following equality
$$
R_{12}^{a,b} \, R_{13}^{a,c} \, R_{23}^{b,c} = R_{23}^{b,c} \, R_{13}^{a,c} \, R_{12}^{a,b} 
$$
holds in $H$.
\end{corollary}

\begin{example}
	Let
	\FloatBarrier
	\begin{figure}[!ht]
	\centering
	\begin{tikzcd}[ampersand replacement=\&]
		1 \ar[r, ""] \& H \ar[r, "i"] \& G \ar[r, "j"] \& K \ar[r, ""] \& 1
	\end{tikzcd}
	\end{figure}
	\FloatBarrier
\noindent	be a group extension with $K$ is a group of exponent 2. As it is well known, in this case $K$ is abelian. For example, one can take  $K \simeq (\Z/2\Z)^n$.
On the group $(G, \circ)$	there exists an elementary quandle solution $R : G^2 \to G^2$ of the YBE, 
	$$
		R(x, y) = (x, x \circ y^{-1} \circ x),
	$$
which corresponds to the core quandle $(Core(G), *)$ with the operation $y * x = x \circ y^{-1} \circ x$.
Considering an elements of $G$ as a set of pairs $(g, k)$, $ g \in H$, $k \in K$, we can express group operations in $G$ in terms of pairs, 
	$$
		(g_1, k_1) \circ (g_2, k_2) = (g_1 \underset{k_1, k_2}{\circ} g_2, k_1 \circ k_2),
	$$
	$$
		(g_1, k_1)^{-1} = (g^{-1}_{k,k^{-1}}, k_1^{-1}),
	$$
	where $\underset{a, b}{\circ} : H \to H$ is an operation on $H$ that depends on two parameters $a, b$, and $g^{-1}_{k,k^{-1}}$ is the inverse element under the operation $\underset{k, k^{-1}}{\circ}$.
	Then the map $R^{a, b} : H^2 \to H^2$,
	$$
		R^{a,b}(x, y) = \left( x, (x \underset{a, b^{-1}}{\circ} y^{-1}_{b^{-1}, b}) \underset{ab^{-1}, a}{\circ} x \right)
	$$
	is a solution of the parametric YBE
	$$
		R_{12}^{a,b}R_{13}^{a,c}R_{23}^{b,c}
		=
		R_{23}^{b,c}R_{13}^{a,c}R_{12}^{a,b}.
	$$
\end{example}

\begin{remark}
As we know, on arbitrary group $G$ one can defines two types of verbal quandles: $Core(G)$ with the operation $g * h = h g^{-1} h$ and $Conj_n(G)$ with the operation $g *_n h = h^{-n} g h^n$ for some integer $n$. We considered $Core(G)$ in the previous example. The case $Conj_1(G)$ was studied in \cite{PT}. By analogy it is  not difficult to describe 
the general case $Conj_n(G)$. In this case we get a solution of the parametric YBE if the subgroup $K^n$ of $K$ which is generated by $n$-th powers which   lies in the center $Z(K)$ of $K$.
\end{remark}

\subsection{Representations of the virtual braid group and the YBE} \label{lin}

In the papers \cite{BN} and \cite{BMN} were constructed some representations of the virtual braid group $VB_n$. In particular, in \cite{BN} was constructed  a  representation
$\varphi :VB_{n}\longrightarrow Aut(F_{n,3n})$  into the automorphism group of the free product 
 $F_{n,3n}=F_{n}\ast \mathbb{Z}^{3n}$, where $F_n= \left< x_{1}\ldots x_{n}\right>$  is the free group of rank $n$ and 
$\mathbb{Z}^{3n}=\left<w_{1},\ldots, w_{n},u_{1}, \ldots, u_{n} ,v_{1}, \ldots, v_{n} \right>$ is the free abelian group of rank $3n$.
This representation is defined by the action on the generators:
$$\begin{matrix}\varphi (\sigma_{i} ):\begin{cases}x_{i}\longrightarrow x_{i}x^{u_{i}}_{i+1}x^{-w_{i+1}u_{i+1}}_{i},&\\ x_{i+1}\longrightarrow x^{w_{i+1}}_{i},&\end{cases} &
\varphi (\sigma_{i} ):\begin{cases}w_{i}\longrightarrow w_{i+1},&\\ w_{i+1}\longrightarrow w_{i},&\end{cases} \\ \\
\varphi (\sigma_{i} ):\begin{cases}u_{i}\longrightarrow u_{i+1},&\\ u_{i+1}\longrightarrow u_{i},&\end{cases} \,\,\,\,\,\,\,\,\,\,\,\,\,\,\,\,\,\,\,\,\,\,\,\,\,\,\,\,\,\,&
\varphi (\sigma_{i} ):\begin{cases}v_{i}\longrightarrow v_{i+1},&\\ v_{i+1}\longrightarrow v_{i},&\end{cases}\,\,\,\,\end{matrix} $$

$$\begin{matrix}\varphi (\rho _{i} ):\begin{cases}x_{i}\longrightarrow x^{v_{i}^{-1}}_{i+1},&\\ x_{i+1}\longrightarrow x^{v_{i+1}}_{i},&\end{cases}  &
 \,\,\,\,\,\, \,\,\,\,\,\, \,\,\,\,\,\, \,\,\,\,\,\,\varphi (\rho _{i} ):\begin{cases}w_{i}\longrightarrow w_{i+1},&\\ w_{i+1}\longrightarrow w_{i},&\end{cases} \\ \\
\varphi (\rho _{i} ):\begin{cases}u_{i}\longrightarrow u_{i+1},&\\ u_{i+1}\longrightarrow u_{i},&\end{cases}  \,\,\,\,\,\,&
 \,\,\,\,\,\, \,\,\,\,\,\,  \,\,\,\,\,\, \,\,\,\varphi (\rho _{i} ):\begin{cases}v_{i}\longrightarrow v_{i+1},&\\ v_{i+1}\longrightarrow v_{i}, &\end{cases} \end{matrix} $$

In this subsection we prove

\begin{theorem} \label{psol}
	Let $G=B\ast A$, where B is a group and A is an abelian group.
	Then the maps
	$$R^{u,v,w}_{12}(x,y,z)=(x^{w},xy^{u}x^{-wv},z),$$
	$$R^{u,p,q}_{13}(x,y,z)=(x^{u},y,xz^{u}x^{-pq}),$$
	$$R^{v,p,q}_{23}(x,y,z)=(x, y^{q},yz^{v}y^{-pq}),$$
gives a solution on $B$ of the  parametric  Yang--Baxter equations:
	$$R^{u,v,w}_{12}R^{u,p,q}_{13}R^{v,p,q}_{23}=R^{v,p,q}_{23}R^{u,p,q}_{13}R^{u,v,w}_{12}$$
with parameters $u,v,w,p,q, u_1,v_1,u_2,v_2 \in A$.

The following maps
$$T^{u,v}_{12}(x,y,z)=(x^{u},y^{v},z),$$
$$T^{u_{1},v_{1}}_{13}(x,y,z)=(x^{u_1},y, z^{v_1}),$$
$$T^{u_{2},v_{2}}_{23}(x,y,z)=(x, y^{u_2},z^{v_2}),$$
gives a solution of  the parametric Yang--Baxter equation:
$$T^{u,v}_{12}T^{u_{1},v_{1}}_{13}T^{u_{2},v_{2}}_{23}=T^{u_{2},v_{2}}_{23}T^{u_{1},v_{1}}_{13}T^{u,v}_{12}$$
on $B$ with parameters in $A$.
\end{theorem}

\begin{proof}
For the group  $G=B\ast A$ consider  the maps
$$R,T : G \times G \to G \times G$$ that are defined  by the formulas:
$$R^{u,v,w}(x,y)=(x^w,xy^{u}x^{-wv}),$$
$$T^{u,v}(x,y)=(x^u,y^{v}), ~~~ x, y \in B,~~u,v,w \in A.$$

We want to find a map $R$ which gives a   solution of the parametric Yang-Baxter equation
$$R^{u,v,w}_{12}R^{u_{1},v_{1},w_{1}}_{13}R^{u_{2},v_{2},w_{2}}_{23}=R^{u_{2},v_{2},w_{2}}_{23}R^{u_{1},v_{1},w_{1}}_{13}R^{u,v,w}_{12},$$
where
$$R^{u,v,w}_{12}(x,y,z)=(x^{w},xy^{u}x^{-wv},z),$$
$$R^{u_{1},v_{1},w_{1}}_{13}(x,y,z)=(x^{w_{1}},y,xz^{u_{1}}x^{-w_{1}v_{1}}),$$
$$R^{u_{2},v_{2},w_{2}}_{23}(x,y,z)=(x, y^{w_{2}},yz^{u_{2}}y^{-w_{2}v_{2}}).$$

Acting by the right side, we get
$$R^{u,v,w}_{12}R^{u_{1},v_{1},w_{1}}_{13}R^{u_{2},v_{2},w_{2}}_{23} (x, y, z)= R^{u,v,w}_{12} R^{u_{1},v_{1},w_{1}}_{13}(x, y^{w_2},yz^{u_2}y^{-w_2v_2})=
$$
$$
=R^{u,v,w}_{12}(x^{w_1}, y^{w_2},x{(yz^{u_2}y^{-w_2v_2})}^{u_1}x^{-w_1v_1})=$$
$$=(x^{w_1w}, x^{w_1}y^{w_{2}u}x^{-w_1wv}, xy^{u_1}z^{u_2u_1}y^{-w_2v_2u_1}x^{-w_1v_1}).$$

Acting by the left side, we get
$$R^{u_{2},v_{2},w_{2}}_{23}R^{u_{1},v_{1},w_{1}}_{13}R^{u,v,w}_{12} (x, y, z) =R^{u_{2},v_{2},w_{2}}_{23} R^{u_{1},v_{1},w_{1}}_{13}(x^{w},xy^{u}x^{-wv},z)=
$$
$$
=R^{u_{2},v_{2},w_{2}}_{23}(x^{ww_1},xy^{u}x^{-wv},x^{w}z^{u_1}x^{-ww_1v_1})=$$
$$=(x^{ww_1},x^{w_2}y^{uw_2}x^{-wvw_2},xy^ux^{-wv+wu_2}z^{u_1u_2}x^{-ww_1v_1u_2+wvw_2v_2}y^{-uw_2v_2}x^{-w_2v_2}).$$

Since $x, y, z$ are arbitrary elements in $B$, we now have the following system of equations on the parameters:
$$
\begin{cases}
	ww_1=w_1w, \, \, \, w_1=w_2,  \, \, \,  uw_2=w_2u,  \, \, \,  wvw_2=w_1wv,
	\\
	u=u_1,\, \, \,  w(u_{2}-v)=0, \, \, \,  u_1u_2=u_2u_1, \, \, \,  w_2v_2=w_1v_1,
	\\
	w(vw_2v_2-w_1v_1u_2)=0,  \, \, \,  uv_2w_2=w_2v_2u_1.
\end{cases}
$$
From this system we obtain:  $w_1=w_2, u=u_1, u_2=v, v_2=v_1$.
It means that  $R^{u,v,w}_{12}$ has three free parameters, $R^{u_{1},v_{1},w_{1}}_{13}=R^{u,v_{1},w_{1}}_{13}$ has two free parameters,
$R^{u_{2},v_{2},w_{2}}_{23}=R^{v,v_{1},w_{1}}_{23}$ does not have  free parameters.

Hence, we get a solution of the following parametric Yang--Baxter equation.
$$R^{u,v,w}_{12}R^{u,p,q}_{13}R^{v,p,q}_{23}=R^{v,p,q}_{23}R^{u,p,q}_{13}R^{u,v,w}_{12}.$$

\medskip

Let us consider  the map $T : G \times G \to G \times G$ that  is  given by the formula $T^{u,v}(x,y)=(x^u,y^{v}).$ Then
$$T^{u,v}_{12}(x,y,z)=(x^{u},y^{v},z),$$
$$T^{u_{1},v_{1}}_{13}(x,y,z)=(x^{u_1},y, z^{v_1}),$$
$$T^{u_{2},v_{2}}_{23}(x,y,z)=(x, y^{u_2},z^{v_2}),$$
are the maps $G^3 \to G^3$. We will show that these maps satisfy some parametric YBE. 

Acting by the right  side, we get
$$T^{u,v}_{12} T^{u_{1},v_{1}}_{13}T^{u_{2},v_{2}}_{23} (x, y, z) =T^{u,v}_{12}T^{u_{1},v_{1}}_{13} (x,y^{u_2},z^{v_2})=T^{u,v}_{12}(x^{u_1},y^{u_2},z^{v_2v_1})=(x^{u_1u},y^{u_2v},z^{v_2v_{1}}).$$
Acting by the  left   side, we get
$$T^{u_{2},v_{2}}_{23}T^{u_{1},v_{1}}_{13}T^{u,v}_{12} (x, y, z) =T^{u_{2},v_{2}}_{23}T^{u_{1},v_{1}}_{13}(x^{u},y^{v},z)=T^{u_{2},v_{2}}_{23} (x^{uu_1},y^{v},z^{v_1})=(x^{uu_1},y^{vu_2},z^{v_1v_2}).$$
Therefore, there are six free parameters.
\end{proof}

\bigskip

\section{Solutions of the parametric {\SE}}

In this section we generalise results of the previous section.

Suppose that, as in Subsection \ref{Gr}, $G$ is an extension of its subgroup $H$ by a group $K$ and 
 $(G, R)$ and $(K, r)$ are solutions of the {\SE} on $G$ and $K$, respectively,  such that the following diagram is commutative
\FloatBarrier
\begin{figure}[!ht]
\centering
\begin{tikzcd}[ampersand replacement=\&]
	G^n  \arrow[r, "j^n"]\arrow[d, "R"] \& K^n \arrow[d, "r"] \\
	G^n \arrow[r, "j^n"] \& K^n
\end{tikzcd}
\end{figure}
\FloatBarrier
\noindent where $G^n$ and $K^n$ are Cartesian products.

Any element $x$ of $G$ can be presented as a pair $(h, k)$, where $k = j(x) \in K$, $h=i^{-1}(x\cdot(\varphi(k))^{-1}) \in H$ where $i$ is an embedding of $H$ and $i(h) = (h, 1)$.
Now we can express 
$$
R(x_1, \dotss , x_n) = R((h_1,k_1), (h_2,k_2), \ldots (h_n, k_n)) = (y_1, y_2, \ldots, y_n) \in G^n.
$$
We can present the last $n$-tuple  in the form
$$
(y_1, y_2, \ldots, y_n) = ((h_1',k_1'), (h_2',k_2'), \ldots (h_n', k_n'))~\mbox{for some}~h_i' \in H, k_i' \in K.
$$
Then we put
$$
R^{k_1,\, \dotss ,k_n}(h_1, \dotss , h_n) = (h_1', h_2', \ldots h_n')
$$
and
$$
r(k_1, \dotss , k_n) = (k_1', k_2', \ldots k_n').
$$
The  solution $R$ can be rewritten as
$$
R(x_1, \dotss , x_n) = (R^{k_1,\, \dotss ,k_n}(h_1, \dotss , h_n), r(k_1, \dotss , k_n)).
$$
Since $R$ is a solution of the {\SE} on $G$, then $R^{k_1,\, \dotss ,k_n}$ is a solution of a parametric {\SE} on $H$. More precisely

\begin{proposition} 
Let  $(G, R)$ be a solution of  the {\SE}. 
Define $$s_{m}(k_1, \ldots, k_N) = p^{\overline m}(r_{\overline{m-1}} \circ \ldots \circ r_{\bar 1}(k_1, \ldots, k_N)),$$
$$z_{m}(k_1, \ldots, k_N) = p^{\overline m}(r_{\overline{m+1}} \circ \ldots \circ r_{\overline{n+1}}(k_1, \ldots , k_N)), $$
where $ p^{\overline m} $  denotes a projection from $K^N$ to $K^{n}$ with a kernel consisting of components which numbers are not in the multi-index $\overline m$. 
Then 
	$(H, R^{k_1,\, \dotss , k_n})$ is a solution of the parametric {\SE}:
	$$
		R_{\overline{1}}^{z_1(k)}
		R_{\overline{2}}^{z_2(k)}
		\cdots
		R_{\overline{n+1}}^{z_{n+1}(k)}
		=
		R_{\overline{n+1}}^{s_{n+1}(k)}
		\cdots
		R_{\overline{2}}^{s_2(k)}
		R_{\overline{1}}^{s_1(k)},
	$$
	where $k = (k_1, ... , k_N) \in K^N$ is a vector of parameters.
\end{proposition}
In the case when  $r =\id$ is  the identity map, we have 
$$
	R_{\bar{1}}^{k_{\bar{1}}}R_{\bar{2}}^{k_{\bar{2}}} \cdots R_{\overline{n+1}}^{k_{\overline{n+1}}}
	=
	R_{\overline{n+1}}^{k_{\overline{n+1}}} \cdots R_{\bar{2}}^{k_{\bar{2}}}	R_{\bar {1}}^{k_{\bar{1}}}.
$$

It is also worth mentioning, that if there is a value $k_0 \in K$ such that $r(k_0, \dotss , k_0) = (k_0, \dotss , k_0)$ then
$(H, R^{k_0,\, \dotss , k_0})$ is a solution of the {\SE}.

\subsection{Solutions of the {\SE} on group extensions}

Let
\FloatBarrier
\begin{figure}[!ht]
	\centering
	\begin{tikzcd}[ampersand replacement=\&]
	1 \ar[r, ""] \& H \ar[r, "i"] \& G \ar[r, "j"] \& K \ar[r, ""] \& 1
	\end{tikzcd}
\end{figure}
\FloatBarrier
\noindent be a group extensions and $(H, R)$, $(K, T)$ be solutions of the {\SE} on $H$ and $K$, respectively.

\begin{definition}
We say that $(G, Q)$ is an extension of the solution $(H, R)$ by the solution $(K, T)$ if it is a solution of the {\SE} and the following diagram is commutative
\FloatBarrier
\begin{figure}[!ht]
\centering
\begin{tikzcd}[ampersand replacement=\&]
	H^n \arrow[r, "i^n"] \arrow[d, "R"]
	\& G^n  \arrow[r, "j^n"]\arrow[d, "Q"] \& K^n \arrow[d, "T"] \\
	H^n \arrow[r, "i^n"]
	\& G^n \arrow[r, "j^n"] \& K^n
\end{tikzcd}
\end{figure}
\FloatBarrier
\end{definition}

\noindent We shall identify the group $H^n$ with  its image in $G^n$ under $i^n$.
\begin{proposition} \label{genext}
	For any  solutions $(H, R)$ and  $(K, T)$ of the {\SE} there is a solution $(G, Q)$ of the {\SE}, which is an extension of  $(H, R)$ by $(K, T)$.
\end{proposition}

\begin{proof}
For a group  extension
	\FloatBarrier
	\begin{figure}[!ht]
	\centering
	\begin{tikzcd}[ampersand replacement=\&]
		1 \ar[r, ""] \& H \ar[r, "i"] \& G \ar[r, "j"] \& K \ar[r, ""] \& 1
	\end{tikzcd}
	\end{figure}
	\FloatBarrier
\noindent we can construct  the corresponding extension of their Cartesian powers,
	\FloatBarrier
	\begin{figure}[!ht]
	\centering
	\begin{tikzcd}[ampersand replacement=\&]
		1 \ar[r, ""] \& H^n \ar[r, "i^n"] \& G^n \ar[r, bend left, "j^n"] \& K^n \ar[l,  bend left, "\varphi"] \ar[r, ""] \& 1.
	\end{tikzcd}
	\end{figure}
	\FloatBarrier

	We can write any element $g$ of $G^n$ uniquely as a product $h_g k_g$ where
 	$$
 		k_g = \varphi \circ j^n(k), \qquad  h_g = g k_g^{-1}, 
	$$
	and $h_g$ is an element of $i(N^n)$. 
	Thus we can consider the underlying set of $G^n$ as the Cartesian product of $H^n$ and $K^n$.
	Now define
	$$
		Q(g) = R(h_g) \varphi \left(T j^n(g) \right).
	$$
	In terms of  the Cartesian product defined above it can be written as $(R, \varphi T j^n(g_k))$.
	Note that $(A \times B, F)$ is a solution for arbitrary sets $A$ and $B$ if and only if  its projections are solutions.
	Hence $Q$ is an extension of a solution if and only if  $\varphi T j^n(g_k)$ is a solution on $G$.
	By taking $\varphi = \tilde{\varphi}^n$ for some section $\tilde{\varphi}$ of $G \overset{j}{\to} K$
	we always obtain an extension since the equation
	$$
		(\tilde{\varphi}^n T j^n)_{\bar 1} \dots (\tilde{\varphi}^n T j^n)_{\overline{n+1}} =
		(\tilde{\varphi}^n T j^n)_{\overline{n+1}} \dots (\tilde{\varphi}^n T j^n)_{\bar 1}
	$$
	 is equivalent to 
 	$$
 		\tilde{\varphi}^N T_{\bar 1} \, j^N \dots \tilde{\varphi}^N T_{\overline{n+1}} \, j^N=
 		\tilde\varphi^N T_{\overline{n+1}} \, j^N\dots \tilde\varphi^N T_{\bar{1}} \, j^N.
	$$
	Since $\tilde\varphi$ and $j$ are bijections between $G$ and $\varphi(G)$.
	The equation above holds if and only if   $(G, T)$ is a solution.
\end{proof}

The next example shows how is it possible to construct  new solutions, using Proposition~\ref{genext}.
\begin{example}
Let us consider the additive group of integers $(\mathbb{Z}, +)$ as an extension of the following form
	\FloatBarrier
	\begin{figure}[!ht]
	\centering
	\begin{tikzcd}[ampersand replacement=\&]
		0 \ar[r, ""] \& \mathbb Z \ar[r, "i"] \& \mathbb Z\ar[r, "j"] \& \mathbb Z_p \ar[r, ""] \& 0.
	\end{tikzcd}
	\end{figure}
	\FloatBarrier
\noindent Here $i$ is the multiplication by $p$ and for any $b \in \mathbb{Z}$ the value $j (b)$ is the reminder of $b$  modulo $p$. For simplicity we will denote $ \bar{b}:= j (b)$ and if $a= (a_1, a_2, \dots, a_n) \in \mathbb{Z}^n $, then $ \bar{a} := (\bar{a}_1, \bar{a}_2, \dots, \bar{a}_n) \in \mathbb{Z}_p^n $.

	Let $R$ and $T$  be solutions of the {\SE} on $\mathbb{Z}$ and $ \mathbb{Z}_p $, respectively. Fix a set-theoretic section $\varphi: \mathbb Z_p \to \mathbb Z$.
	Each element $b \in \mathbb{Z}$ can be uniquely presented as
	$$
	(b -  \varphi(\bar{b}), \ \bar{b}) \in \mathbb{Z} \times \mathbb{Z}_p.
	$$
	Then a map
	$$
		a \mapsto R( a- \varphi^{ n}( \bar a )) + \varphi^{ n}T( \bar a)
	$$
	is a solution of the {\SE} over $\mathbb{Z}$.
	
	For example, let $n = 3$, $R = \text{id}$ and $T: (x,y,z) \mapsto (x + 2y - 2z, 2z - y,  z)$ are solutions of the TE on $\mathbb{Z}$ and $\mathbb{Z}_p$, respectively.
	Then the map
	$$
		(x, y, z) \mapsto \left(\overline{(x+2y-2z)}+ x - \bar x, \overline{(2z - y)} +y - \bar y, z \right)
	$$
	gives  a new solution of the  TE on $\mathbb{Z}$, where a section embeds $\Z_p$ in $\Z$ as the first $p$ non-negative elements.
\end{example}


\bigskip

\section{Inverse limits of  solutions} \label{Inlim}

Let $ \mathcal{A} $ be a category of algebraic systems in which {\SE} can be defined. For example, category of groups, rings, modules,  etc.
Define a category $\mathcal{CA}(n)$ of pairs $(X, R_X)$
where $X$ is an object in $\mathcal A$ and $R_X$ is a solution of the {\SE}. 
Morphisms in $\mathcal{CA}(n)$ are morphisms of $ \mathcal{A} $  with an additional condition,
$$
	Mor((X, R_X), (Y, R_Y)) = \{f \in Mor(X, Y) \ | \ f \circ R_X = R_Y \circ f \}.
$$

We show how to define a solution on inverse limits. Definition of an inverse limit and its properties can be found in  \cite[Chapter 3]{Ml}. 
Let $D$ be a small category and $\mathcal F: D \to \mathcal{CA}(n)$ be a diagram.
By applying a forgetful functor $T: \mathcal{CA}(n) \to \mathcal A$ we get a diagram in $\mathcal A$.

\begin{proposition} \label{lim}
	If there is an inverse limit $\lim(T \circ \mathcal F)$ in $\mathcal A$, then there is an inverse limit $\lim \mathcal F$ in $\mathcal{CA}(n)$.
\end{proposition}
\begin{proof}
	Since limits commute with limits, we have the equality 
$$\lim\limits_{a \in D} \mathcal F(a)^n = (\lim\limits_{a \in D} \mathcal F(a))^n.
$$
	Let $X = \lim(T \circ \mathcal F)^n$, $F$ is an $\mathcal F^n$. For an object $A \in D$ let $p_A \in \text{Mor}(X, F(A))$ be a morphism from the definition of a limit.
	We define a map $R$ on $X$,
	$$
		(x_A)_{A \in D} \mapsto (R_A(x_A))_{A \in D}.
	$$
	First we need to show that $R (X) \subset X$.
	Suppose there is an $x \in X$ such that $R(x) \not\in X$,
	hence there is a morphism $f_{AB}$, such that \[f_{AB} \circ p_A(R(x)) \neq p_B(R(x))\]\\
	which is equivalent to 
	$$
		f_{AB} \circ R_A(x_A) \neq R_B(x_B).
	$$
	By the condition on morphisms in this category we have
	$$
		R_B(f_{AB}(x_A))\neq R_B(x_B).
	$$
	Since $f(x_A) = x_B$ we get a contradiction.
		
	For any objects $A$, $B \in D$ and a morphism $f \in \text{Mor}(A, B)$ we have a commutative diagram
	\FloatBarrier
	\begin{figure}[!ht]
	\centering
	\begin{tikzcd}[ampersand replacement=\&]
		\& X \ar[ld, swap, "p_{A}"]  \ar[rd, "p_{B}"] \& \\
		F(A) \ar[rr, "f_{AB}"] \& \& F(B)
	\end{tikzcd}
	\end{figure}
	\FloatBarrier
\noindent	where $f_{AB} = F(f)$.
\\
	Since $x_{A} = p_{A}(x)$ by the condition on morphisms  we obtain
	$$
		f_{AB} \circ R_{A} \circ p_{A}(x) = R_{B} \circ f_{AB} \circ p_{A}(x) = R_{B} \circ p_{B}(x).
	$$
	Hence the following diagram is commutative
	\FloatBarrier
	\begin{figure}[!ht]
	\centering
	\begin{tikzcd}[ampersand replacement=\&]
		\& (X, R) \ar[ld, swap, "p_{A}"]  \ar[rd, "p_{B}"] \& \\
		(F(A), R_{A}) \ar[rr, "f_{AB}"] \& \& (F(B), R_{B})
	\end{tikzcd}
	\end{figure}
	\FloatBarrier
\end{proof}

\begin{corollary} \label{adic}
Consider $p$-adic integers $\mathbf{Z}_p$ as an inverse limit of $ \Z_{p^k} $ and suppose that 
$(\Z_{p^k}, R_k)$, $k=1, 2, \ldots$ are solution  of the {\SE}, such that $R_k \circ (\phi_k)^{\times n} = R_{k-1}$, where $\phi_k$ is a projection $\Z_{p^k} \to \Z_{p^{k-1}}$. Then this family defines a solution $(\mathbf{Z}_p, R)$ of the {\SE}.
\end{corollary}

\begin{example} Take a family of linear solutions of the TE on $\Z_{p^k}$,
$$R_k(x, y, z) = (a_k x, (1-a_k b_k)x + b_k y, (1-a_k b_k)c_k x + (1-b_k c_k)y + c_k z),$$
and suppose that the coefficients satisfy the conditions
$$a_{k+1} \text{ mod }  p^{k} = a_{k}, \quad b_{k+1} \text{ mod } p^{k} = b_{k},  \quad c_{k+1} \text{ mod } p^{k} = c_{k}.$$
Then this  family $R_k$ gives a solution $R$ on $\mathbf Z_{p}$,
	$$ R(x, y, z) = (a x, (1-a b)x + b y, (1-a b)c x + (1-b c)y + c z), $$
	where $a$, $b$ and $c$ are $p$-adic numbers defined by the sequences
	$$a = (a_1, \ldots, a_k, \ldots), \quad b = (b_1, \ldots, b_k, \ldots), \quad c = (c_1, \ldots, c_k, \ldots).$$
\end{example}



\bigskip

\section{Tetrahedron equation and corresponding algebraic systems } \label{algsust}

In this section we shall consider  the tetrahedron equation (TE),
$$
	R_{123} R_{145} R_{246} R_{356} = R_{356} R_{246} R_{145} R_{123},
$$
 and by solution we mean a set-theoretical solution of this equation.

\subsection{3-ternoids}
An algebraic system with one ternary operation is said to be  {\it ternar},
an algebraic system with $k$ ternary operations is said to be  $k$-{\it ternoid}. 
If $R  = (f, g, h) : X^3 \to X^3$ is a solution of the TE on some set $X$,
then we can define on $X$ three ternary operations $[\cdot, \cdot, \cdot]_i : X^3 \to X$, $i = 1, 2, 3$, by the rules
$$
	[a, b, c]_1 = f(a, b, c),~~~[a, b, c]_2 = g(a, b, c),~~~[a, b, c]_3 = h(a, b, c),~~~a, b, c \in X.
$$
For simplicity, we shall  denote the different operations by different type of brackets and write 
$$
[a, b, c]_1 = [a, b, c],~~~[a, b, c]_2 = \lg a, b, c \rg,~~~[a, b, c]_3 = \{a, b, c\}. 
$$


Using direct calculations we get

\begin{proposition}
\label{p7.1}
	Let $(X, [\cdot, \cdot, \cdot], \lg \cdot, \cdot, \cdot \rg, \{ \cdot, \cdot, \cdot \})$ be a 3-ternoid.
	Then it defines a solution of the TE if and only if the following six equalities hold
	$$
	\begin{gathered}~
		 [[ x, \lg y, t, \{z, p, q \} \rg, \lg z, p, q \rg ], [y, t, \{z, p, q \}], [z, p, q]] =
		 [[x, y, z], t, p],
		 \\
		 \lg [x, \lg y, t, \{z, p, q \} \rg, \lg z, p, q \rg], [y, t, \{z, p, q \}], [z, p, q]\rg =
		 [\lg x, y, z\rg, \lg [x, y, z], t, p\rg, q],   
		 \\
		 \{[x, \lg y, t, \{z, p, q \} \rg, \lg z, p, q \rg], [y, t, \{z, p, q \}], [z, p, q]\} =
		 \\=
		 [ \{x, y, z\}, \{[x, y, z], t, p\}, \{\lg x, y, z\rg,  \lg [x, y, z], t, p\rg, q\}],  
		 \\
		 \lg x, \lg y, t, \{z, p, q \} \rg, \lg z, p, q \rg \rg =
		 \lg \lg x, y, z \rg, \lg [x, y, z], t, p\rg, q\rg,
		 \\
		 \{x, \lg y, t, \{z, p, q \} \rg, \lg z, p, q \rg \} =
		 \lg \{x, y, z\}, \{[x, y, z], t, p\}, \{\lg x, y, z\rg, \lg [x, y, z], t, p\rg, q\}\rg,
		 \\
		 \{y, t, \{z, p, q \}\} =
		 \{ \{x, y, z\},  \{[x, y, z], t, p\}, \{\lg x, y, z\rg, \lg [x, y, z], t, p\rg, q\}\}
	\end{gathered}
	$$
	for all $(x, y, z, t, p, q) \in X^6$.
\end{proposition}

The next three corollaries describe elementary solutions of the TE.

\begin{corollary}
\label{c1}
	Let $(X, [\cdot, \cdot, \cdot])$ be a ternar.
The map $R : X^3 \to X^3$,
	$$
		R(a,b,c) = ([a, b, c], b, c),~~~a, b, c \in X,
	$$
satisfies the TE if and only if 
	$$
		[[x, t,  p], [y, t,  q], [z, p, q]] = [[x, y, z], t, p],~~\mbox{for all}~x, y, z, t, p, q \in X.
	$$
\end{corollary}

\begin{corollary}
\label{c2}
	Let $(X, [\cdot, \cdot, \cdot])$ be a ternar.
The map $R : X^3 \to X^3$,
	$$
		R(a,b,c) = (a, [a, b, c], c),~~~a, b, c \in X.
	$$
satisfies the TE if and only if
	\begin{equation}
	\label{e2}
		 [[x, y,  z], [x, t,  p], q] = [x, [y, t, q], [z, p, q]],~~\mbox{for all}~x, y, z, t, p, q \in X.
	\end{equation}
\end{corollary}

\begin{corollary}
\label{c3}
	Let $(X, [\cdot, \cdot, \cdot])$ be a ternar.
The map $R : X^3 \to X^3$,
	$$
		R(a,b,c) = (a, b, [a, b, c]),~~~a, b, c \in X.
	$$
satisfies the TE if and only if 
	\begin{equation}
	\label{e3}
		[[x, y,  z], [x, t,  p], [y,t, q]] = [y, t, [z, p, q]],~~\mbox{for all}~x, y, z, t, p, q \in X.
	\end{equation}
\end{corollary}

\begin{remark}
One can see that the identity in Corollary \ref{c2} is more symmetric than identities in Corollary \ref{c1} and in Corollary \ref{c3}. In particular, the number of variables in the left side of this identity is the same as in the right side. In two other identities these numbers are different.
\end{remark}

We will call a solution from Corollary \ref{c1} an {\it elementary solution of the first type} or, simply elementary $1$-solution; a solution from Corollary \ref{c2} an {\it elementary solution of the second type} or simply elementary $2$-solution; a solution from Corollary \ref{c3} an {\it elementary solution of the third type} or, simply elementary $3$-solution.

Let $P_{13} : X^3 \to X^3$, $P_{13}(x, y, z) = (z, y, x)$, be the permutation of the first and the third components. If $R$ is an elementary 1-solution, then $P_{13} R P_{13}$ is an elementary 3-solution. Hence, we have to study only elementary 1- and  2-solutions.

\subsection{Elementary 1-solutions}

For studying elementary 1-solutions of the TE let us introduce 
 an algebraic system $ (X, \os, \oc, \oL, \oR) $ with four binary operations which satisfy axioms:
$$
\begin{gathered}
	x \oc y = (x \oL z) \oc (y \oL z);
	\\
	(x \oc y) \os (z \oc w) = (x \os z) \oc (y \os w);
	\\
	(x \oR y) \oR z = (x \oR z) \oR (y \os z);
	\\
	(x \os y) \oL z = x \oR (y \oc z).
\end{gathered}
$$
We call it first tetrahedral 4-groupoid and denote it by $T_1$-groupoid.

\begin{proposition}  \label{1-sol}
	$T_1$-groupoid gives an elementary 1-solution $(X, R)$ of the TE if we put
	$$
		R(x, y, z) = (x \oR (y \oc z), \; y, \; z),~~~x, y, z \in X.
	$$ 
\end{proposition}
\begin{proof}
	In order for $R$ to be a solution it shall satisfy equality from corollary \ref{c1}.
	The equality then takes the form
	$$
	(x \oR (t \oc p)) \oR ((y \oR (t \oc q)) \oc (z \oR (p \oc q))) =
	(x \oR (y \oc z)) \oR (t \oc p).
	$$
	Applying axioms to the left side of this equality we get
	$$
	\begin{gathered}
	(x \oR (t \oc p)) \oR ((y \oR (t \oc q)) \oc (z \oR (p \oc q))) =
	(x \oR (t \oc p)) \oR (((y \os t) \oL q) \oc ((z \os p) \oL q)) =
	\\=
	(x \oR (t \oc p)) \oR ((y \os t) \oc (z \os p)) =
	(x \oR (t \oc p)) \oR ((y \oc z) \os (t \oc p)) =
	(x \oR (y \oc z)) \oR (t \oc p).
	\end{gathered}
	$$
\end{proof}

\begin{example}\label{e7}
	Consider an algebraic system $ (V, \os, \oc, \oL, \oR) $ on a vector space $ V $
	with operations defined as follows:
	$$
	\begin{aligned}
		x \os y & := \b x + (1 - \b) y, \\
		x \oc y & := x - y, \\
		x \oL y & := x + (\b - 1) y, \\
		x \oR y & := \b x + (1 - \b) y,
	\end{aligned}
	$$
	where $\b$ is some  endomorphism of $ V $.
	From this system we get the solutions 
	$$
R(x, y, z) = (\b x + (1 - \b) y + (\b - 1) z, y, z),~~~x, y, z \in V.
$$
\end{example}

There are 2-groupoids, that give elementary 1-solutions.
Consider for example a system $(X, \os, \oc)$ with the following axioms:
$$
\begin{gathered}
	x \oc y = (x \oc z) \oc (y \oc z);
	\\
	(x \os y) \os z = (x \os z) \os (y \os z);
	\\
	(x \oc y) \os (z \oc w) = (x \os z) \oc (y \os w).
\end{gathered}
$$
We will call it a reduced first tetrahedral 4-groupoid or simply reduced $T_1$-groupoid.
The following claim holds
\begin{proposition}
Reduced $T_1$-groupoid gives an elementary 1-solution $(X, R)$ of the  TE, where
	$$
		R(x, y, z) = (x \os (y \oc z), \; y, \; z),~~~x, y, z \in X.
	$$ 
\end{proposition}
\begin{proof}
In order for $R$ to be a solution it shall satisfy the equality from corollary \ref{c1} which takes the form
	$$
	\begin{gathered}
	(x \os (t \oc p)) \os ((y \os (t \oc q)) \oc (z \os (p \oc q))) =
	(x \os (y \oc z)) \os (t \oc p).
	\end{gathered}
	$$
	Applying the axioms to the left side of this equality we get
	$$
	\begin{gathered}
	(x \os (t \oc p)) \os ((y \os (t \oc q)) \oc (z \os (p \oc q))) =
	(x \os (t \oc p)) \os ((y \oc z) \os ((t \oc q) \oc (p \oc q))) =
	\\=
	(x \os (t \oc p)) \os ((y \oc z) \os (t \oc p)) =
	(x \os (y \oc z)) \os (t \oc p).
	\end{gathered}
	$$
\end{proof}

\begin{example}
	Consider an algebraic system $ (V, \os, \oc) $ on a vector space $ V $
	with operations defined as follows:
	$$
	\begin{aligned}
		x \os y & := \b x + (1 - \b) y, \\
		x \oc y & := x - y,
	\end{aligned}
	$$
	where $\b$ is some fixed endomorphism of  $ V $.
	From this system we get the elementary 1-solution 
$$
R(x, y, z) = (\b x + (1 - \b) y + (\b - 1) z, y, z),~~~x, y, z \in V.
$$
\end{example}
	Note that this system is very similar to system from example \ref{e7}. It is generally true that if in $T_1$-groupoid
	$ x \os y = x \oR y $ then forgetting about operations $ x \oL y $ and $ x \oR y $ yields a reduced $T_1$-groupoid.
	It is not clear if there is a $T_1$-groupoid such that forgetting about $ x \oL y $ and $ x \oR y $ will not yield a reduced $T_1$-groupoid.

\subsection{Elementary 2-solutions} For studying elementary 2-solutions of the TE let us introduce an algebraic system $ (X, \ms, \mc, \mL, \mR)$ with 4 binary operations which satisfies the following axioms:
$$
\begin{gathered}
	x \mR (y \ms z) = (x \mR y) \ms (x \mR z),
	\\
	(x \mc y) \mL z = (x \mL z) \mc (y \mL z),
	\\
	(x \ms y) \mc (z \ms w) = (x \mc z) \ms (y \mc w),
	\\
	(x \mR y) \mL z = x \mR (y \mL z),
	\\
	(x \ms y) \mL z = x \mR (y \mc z).
\end{gathered}
$$

Further we shall call this algebraic system by second tetrahedral 4-groupoid and denote it by $T_2$-groupoid. 

\begin{proposition} \label{2-sol}
	Any $T_2$-groupoid $X$ defines  an elementary 2-solution $(X, R)$ of the TE by the formula
	$$
		R(x, y, z) = (x, \; x \mR (y \mc z), \; z),~~~x, y, z \in X.
	$$ 

\end{proposition}
\begin{proof}
	In order for $R$ to be a solution it shall satisfy equality from Corollary \ref{c2}.
	This equality then takes the form
	$$
	x \mR ((y \mR (t \mc q)) \mc (z \mR (p \mc q))) =
	(x \mR (y \mc z)) \mR ((x \mR (t \mc p)) \mc q).
	$$
	Applying the axioms of the $T_2$-groupoid  to the left side of this equality we get
	$$
	\begin{gathered}
	x \mR ((y \mR (t \mc q)) \mc (z \mR (p \mc q))) =
	x \mR (((y \ms t) \mL q) \mc ((z \ms p) \mL q)) =
	\\=
	x \mR (((y \ms t) \mc (z \ms p)) \mL q) =
	(x \mR ((y \mc z) \ms (t \mc p))) \mL q =
	\\=
	((x \mR (y \mc z)) \ms (x \mR (t \mc p))) \mL q =
	(x \mR (y \mc z)) \mR ((x \mR (t \mc p)) \mc q).
	\end{gathered}
	$$
\end{proof}

Now suppose that we have an elementary 2-solution $(X, R)$ of the TE, where
$$
	R(x, y, z) = (x, [x, y, z], z),~~~x, y, z \in X.
$$
We can formulate a question: Is it defines a $T_2$-groupoid ?

From the next proposition follows that under some assumptions on the ternar $(X, [\cdot, \cdot, \cdot])$ the answer is positive. 

\begin{proposition}
Suppose that a ternar $(X, [\cdot, \cdot, \cdot])$ defines a solution of the TE, and  there exist an element $c \in X$ such that $[c, c, c] = c$, and a function $\{\cdot\} : X \to X$ such that
$$
	\{[c, x, c]\} = [c, \{x\}, c] = x, \quad [\{x\}, \{y\}, c] = \{[x, y, c]\} \quad \text{ and } \quad [c, \{x\}, \{y\}] = \{[c, x, y]\}.
$$
If we put
$$
	x \ms y = [x, y, c],
	\qquad
	x \mc y = [c, x, y],
	\qquad
	x \mR y = [x, \{y\}, c],
	\qquad
	x \mL y = [c, \{x\}, y],
$$
for any $x, y \in X$, then $(X, \ms, \mc, \mR, \mL)$ is a $T_2$-groupoid.
\end{proposition}

\begin{proof}

It is easy to see that
$$
\begin{gathered}
	\shoveleft{ (x \ms y) \mc (z \ms w) = [c, [x, y, c], [z, w, c]] = \qquad\qquad}\\
	\shoveright{\qquad\qquad = [[c, x, z], [c, y, w], c] = (x \mc z) \ms (y \mc w);}
	\\
	\shoveleft{ x \mR (y \ms z) = [x, \{[y, z, c]\}, c] = [x, [\{y\}, \{z\}, c], [c, c, c]] = \qquad\qquad}\\
	\shoveright{\qquad\qquad = [[x, \{y\}, c], [x, \{z\}, c], c] = (x \mR y) \ms (x \mR z);}
	\\
	\shoveleft{ (x \mc y) \mL z = [c, \{[c, x, y]\}, z] = [[c, c, c], [c, \{x\}, \{y\}], z] = \qquad\qquad}\\
	\shoveright{\qquad\qquad = [c, [c, \{x\}, z], [c, \{y\}, z]] = (x \mL z) \mc (y \mL z);}
	\\
	\shoveleft{ (x \mR y) \mL z = [c, \{[x, \{y\}, c]\}, z] = \{[c, [x, \{y\}, c], [c, z, c]]\} = \qquad\qquad}\\
	\shoveright{\qquad\qquad = \{[[c, x, c], [c, \{y\}, z], c]\} = [x, \{[c, \{y\}, z]\}, c] = x \mR (y \mL z);}
	\\
	\shoveleft{ x \mR (y \mc z) = [x, \{[c, y, z]\}, c] = \{[[c, x, c], [c, y, z], c]\} = \qquad\qquad}\\
	\shoveright{\qquad\qquad = \{[c, [x, y, c], [c, z, c]]\} = [c, \{[x, y, c]\}, z] = (x \ms y) \mL z,}
\end{gathered}
$$
it means that all axioms hold and we have a $T_2$-groupoid.
\end{proof}

Hence, using the operations, constructed in the last proposition, one can define the map
$$ 
(x, y, z) \mapsto (x,  \mR (y \mc z), z) ,~~~x, y, z \in X,
$$ 
which gives a solution of the TE.

\begin{question}
	Is it true that this solution is the same as original one,
by the other words, is it true  that $ [x, y, z] = x \mR (y \mc z) $ for all $x, y, z \in X$?
\end{question}

\begin{example}
	Consider an algebraic system $ (V, \ms, \mc, \mL, \mR) $, where $V$ is a vector space and other  operations are defined as follows:
	$$
	\begin{aligned}
		x \ms y & := (1 - \b) x + \b y, \\
		x \mc y & := \b x + (1 - \b) y, \\
		x \mL y & := (1 - \b) x + y, \\
		x \mR y & := x + (1 - \b) y,
	\end{aligned}
	$$
	where $\b$ is some fixed endomorphism of  $ V $. One can check that this algebraic system is a $T_2$-groupoid and we get an elementary 2-solution $(V, R)$ of the TE if we take
$$R(x, y, z) = (x, (1-\b) x + \b y + (1-\b) z, z),~~~x, y, z \in V.$$
On the other side,  if $\b$ is an automorphism then by taking $c : = 0$ and $ \{ x \} := \b^{-1}x $ we can extract the  $T_2$-groupoid  $ (V, \ms, \mc, \mL, \mR) $ from the solution $(V, R)$.
\end{example}
\begin{remark} It is easy to see that in this $T_2$-groupoid   $ (V, \ms, \mc, \mL, \mR) $ the algebraic systems $(V, \ms)$	and $(V, \mc)$ are Alexander quandles.
\end{remark}


\bigskip

\section{Verbal solutions} \label{verb}

	Let $G$ be a group.
A {\it verbal solution} $(G, R)$ of the {\SE} is  a solution of the form
$$ 
R(g_1, g_2, \ldots, g_n) = (w_1(g_1, g_2, \ldots, g_n),  w_2(g_1, g_2, \ldots, g_n), \ldots, w_n(g_1, g_2, \ldots, g_n)),
$$
where $w_i = w_i(x_1, x_2, \ldots, x_n)$, $i = 1, 2, \ldots, n$, 
are reduced words  in the free group $F_n= \langle x_1, \dotss , x_n \rangle$.
A verbal solution $R$ of the {\SE} is said to be verbal {\it elementary $k$-solution} if each word $w_i$ is equal to $x_i$ with the exception of $w_k$.
	
\subsection{Universal verbal solutions}
If $(G, R)$ is a verbal  solution of the {\SE} and $\varphi : G \to H$ is a homomorphism, such that
$(Ker (\varphi), R|_{Ker (\varphi)})$ is a solution of the {\SE},
then it induces a solution $(\varphi(G), R^{\varphi})$, where 
$$
R^{\varphi}(\varphi(g_1), \varphi(g_2), \ldots, \varphi(g_n)) = (\varphi(g'_1), \varphi(g'_2), \ldots, \varphi(g'_n)),
$$
and
$$
R(g_1, g_2, \ldots, g_n) = (g'_1, g'_2, \ldots, g'_n).
$$

As we know, for arbitrary group $G$ there is a map $R : G^2 \to G^2$ what is an elementary solution to the YBE. For example, we can take some quandle on $G$, conjugacy quandle $Conj(G)$, or core quandle $Core(G)$ and construct elementary solution on $G$. In general case, we can formulate 

\begin{question}
Let $G$ be a  group. Is there a map $R : G^n \to G^n$, $n > 2$, that gives a bijective non-trivial (elementary) solution to the  {\SE}?
\end{question}
By a trivial solution we mean a permutation of components.

\medskip

In this subsection we are studying verbal elementary solutions of the TE on arbitrary group that is equivalent to study solutions on free non-abelian group $F_6$. At first recall some definitions from combinatorial group theory. We will consider $F_n$ as the free product of $n$ infinite cyclic groups, and present elements of $F_n$ as reduced words
\begin{equation} \label{word}
w=x_{i_1}^{\alpha_1} x_{i_2}^{\alpha_2}\ldots x_{i_k}^{\alpha_k},~~\alpha_j \in \mathbb{Z} \setminus \{ 0 \},~~i_j \not= i_{j+1}~\mbox{for}~j = 1, 2,\ldots,k-1.
\end{equation}
Any subword $x_{i_j}^{\alpha_j}$ is called a {\it syllable}. The number $k$ is called the syllable length of $w$ and is denoted by $l(w)$. The $i$-{\it  syllable length} of $w$ is the number of syllables which lie in $\langle x_i \rangle$, it is denoted by $l_i(w)$. 
For example, if $w = x_3^{-5} x_1^2 x_2^{-7} x_1 x_3^8 \in F_3$, then $l(w) = 5$,  $l_1(w) = l_3(w) = 2$, $l_2(w) = 1$.

Also we will use the cyclically reduced form of words. A word (\ref{word}) is called {\it cyclically reduced} if $i_1 \not= i_k$, or if $i_1 = i_k$, then $\alpha_1 \alpha_k > 0$.
If $w$ is not cyclically reduced, then $w \equiv u^{-1}  w_0 u$, where $w_0$ is a cyclically reduced subword of $w$ and $\equiv$ means equality of words (letter by letter). Then $w^m \equiv u^{-1}  w_0^m u$ for any integer $m$. For example, $x_3^{-5} x_1^2 x_2^{-7} x_1 x_3^8 = x_3^{-5} (x_1^2 x_2^{-7} x_1 x_3^3) x_3^5$ and the subword $x_1^2 x_2^{-7} x_1 x_3^3$ is cyclically reduced.

\begin{lemma}\label{reducedLemma}
	Let $w = w(x_1,\ldots, x_n)$ be a reduced word in the free group $F_n$, and $l_j(w)=k > 0$ for some $j$, $1\leq j \leq n$. Then $l_j(w^m)\ge k$ for any integer $m\ne 0$. Moreover, the first and the last symbols of $w$ coincide with the first and the last symbols of $w^m$, respectively, for all positive integer $m$.
\end{lemma}
\begin{proof}
We will assume that $m >0$. The proof for the case $m <0$ is similar.

Suppose at first that $w$ is cyclically reduced. If $l(w) = l_j(w)= 1$, then $w = x_j^{\alpha}$ for some non-zero integer $\alpha$ and the need assertion is true. Suppose that $l(w) > 1$ and $w$ has a form (\ref{word}). If $i_1$ or $i_k$ is different from $j$, then $l_j(w^m)=km$. If $i_1= i_k=j$,
then $k>1$ and 
$$
l_j(w^m)=km-(m-1) = m(k-1) + 1.
$$ 
Since $m \geq 2$, then $m(k-1) + 1 \geq 2(k-1)+1= 2k-1 \geq k$.

If $w$ is not cyclically reduced, then $w \equiv u^{-1}  w_0 u$. If $k = l_j(w) = 2l_j(u) + l_j(w_0)$, then $ l_j(w^m) = 2l_j(u) + l_j(w_0^m)$. By the result above $l_j(w_0^m) \geq l_j(w_0)$. Hence, $ l_j(w^m) \geq 2l_j(u) + l_j(w_0) \geq k$. If $k = l_j(w) = 2l_j(u) + l_j(w_0)-1$ that it is possible if the last syllable of $u^{-1}$ and the first syllable of $w_0$ lies in $\langle x_j \rangle$, or the last syllable of $w_0$ and the first syllable of $u$ lies in $\langle x_j \rangle$ (notice that both possibilities can not be realized). In this case 
$$
l_j(w^m) = 2l_j(u) + l_j(w_0^m)-1 \geq 2l_j(u) + l_j(w_0)-1 \geq k.
$$

The second part of the lemma is evident.

\end{proof}

\begin{lemma}\label{dependLemma}
	Let $w = w(x_1, x_2)$ and $g=g(x_1, x_2, x_3)$ be  reduced words in the free group $F_5$, such that  $l(w)=k$,  $l_3(g)=m>0$, $l(g)>1$. If $n=l_3(w(g(x_1,x_2,x_3),g(x_4,x_5,x_3)))$, then 
	the following inequalities hold:
	\begin{align}
	&n\ge m, \quad &&\text{if} ~k=1, \label{ineq1}\\
	&n\ge 2(m-1) + (m-2)(k-2), \quad &&\text{if }k\ge 2,~m\ge2. \label{ineq3}
	\end{align}
	In the case $k\ge2$ and $g(x_1,x_2,x_3)=g_1(x_1,x_2)x_3^{\a}g_2(x_1,x_2)$, if $g_1$ and $g_2$ are non-trivial simultaneously then $n=k$.
\end{lemma}
\begin{proof}
Suppose that
\begin{equation} \label{word1}
w=x_{1}^{\alpha_1} x_{2}^{\beta_1}  x_{1}^{\alpha_2} x_{2}^{\beta_2}
\ldots x_{1}^{\alpha_s} x_{2}^{\beta_s},~~\alpha_j,  \b_j \in \mathbb{Z},~\mbox{and only}~\a_1\,  \mbox{and/or}~ \b_s~\mbox{can be zero}.
\end{equation}

	If $k=1$, then $w$ is a power of $x_1$ or a power of $x_2$. In this case  the statement follows from Lemma \ref{reducedLemma}.

Let us prove 	 (\ref{ineq3}). If $m\ge2$, then in the word $g(x_1,x_2,x_3)^a \, g(x_4,x_5,x_3)^b$ cancellation of $x_3$ can occur only once, therefore, 
$$
l_3(g(x_1,x_2,x_3)^a \, g(x_4,x_5,x_3)^b) \ge (m - 1) + (m - 1).
$$
 Further, in the word $g(x_1,x_2,x_3)^a \, g(x_4,x_5,x_3)^b \, g(x_1,x_2,x_3)^c$ cancellation of  $x_3$ can occur only twice, therefore 
 $$
 l_3(g(x_1,x_2,x_3)^a \, g(x_4,x_5,x_3)^b \, g(x_1,x_2,x_3)^c)\ge (m - 1) + (m - 2) + (m - 1).
 $$
  Applying this argument $k$ times we get (\ref{ineq3}).
\end{proof}

Using the notations of the previous lemma we get

\begin{lemma}\label{dependLemma2}
	Let $w = w(x_1, x_2)$, $l(w)=k \ge 2$  and 
$$
g(x_1,x_2,x_3)\equiv g_1(x_1,x_2) \, x_3^{\a} \, g_2(x_1,x_2),~~ g_i \not\equiv 1,~~\alpha \not= 0.
$$
 If $n=l_3(w(g(x_1,x_2,x_3),g(x_4,x_5,x_3)))$, then $n=k$.
\end{lemma}

\begin{proof}
Suppose that
\begin{equation} \label{word12}
w=x_{1}^{\alpha_1} x_{2}^{\beta_1}  x_{1}^{\alpha_2} x_{2}^{\beta_2}
\ldots x_{1}^{\alpha_s} x_{2}^{\beta_s},~~\alpha_j,  \b_j \in \mathbb{Z},~\mbox{and only}~\a_1\,  \mbox{and/or}~ \b_s~\mbox{can be zero}.
\end{equation}
Since the words $g_1$ and $g_2$ are not trivial, then 
$$
w(g(x_1,x_2,x_3), g(x_4,x_5,x_3)) = w(g_1(x_1,x_2) \, x_3^{\alpha} \, g_2(x_1,x_2), g_1(x_4,x_5) \, x_3^{\alpha} \, g_2(x_4,x_5)) =
$$	
$$
=\left(g_1(x_1,x_2) \, x_3^{\alpha} \, g_2(x_1,x_2) \right)^{\a_1}  \, (g_1(x_4,x_5) \, x_3^{\alpha} \, g_2(x_4,x_5))^{\b_1} \, (g_1(x_1,x_2) \, x_3^{\alpha} \, g_2(x_1,x_2))^{\a_2}  \cdot
$$
$$
\cdot (g_1(x_4,x_5) \, x_3^{\alpha} \, g_2(x_4,x_5))^{\b_2} \ldots (g_1(x_1,x_2) \, x_3^{\alpha} \, g_2(x_1,x_2))^{\a_s}  \, (g_1(x_4,x_5) \, x_3^{\alpha} \, g_2(x_4,x_5))^{\b_s}.
$$	
We can see that cancellations are possible only inside the words 
$$
\left(g_1(x_1,x_2) \, x_3^{\alpha} \, g_2(x_1,x_2) \right)^{\a_i}~  \mbox{or}~ (g_1(x_4,x_5) \, x_3^{\alpha} \, g_2(x_4,x_5))^{\b_i},
$$
but not possible between these words.
Therefore $n=k$.

\end{proof}

A description of verbal elementary  3-solutions to the TE gives

\begin{theorem}
	Let $R\colon G^3\to G^3$ be a verbal  elementary  3-solution of the TE for every group $G$, then one of following cases holds:
	\begin{enumerate}
		\item $R(x,y,z)=(x,y,yx^{-1})$,
		\item $R(x,y,z)=(x,y,x^{-1}y)$,
		\item $R(x,y,z)=(x,y,w(y,z))$, where $R'(y,z)=(y,w(y,z))$ is a solution of the YBE for every group $G$.
	\end{enumerate}
\end{theorem}

\begin{proof}
	If 
$$
R(x,y,z)= (x,y,w(x,y,z)),~~x, y, z \in G,
$$
 is a solution of the TE, then by Corollary \ref{c1}
	\begin{align}
	\label{eqWord}
	w(w(x,y,z), w(x,t,p), w(y,t,q))=w(y, t, w(z,p,q)).
	\end{align}
	
	Suppose at first, that $w$ does not depend on $x$, i.e.  $w(x,y,z)=w(y,z)$.	In this case the equality (\ref{eqWord}) takes the form
	$$w(w(t,p), w(t,q))=w(t, w(p,q)).$$ 
	Therefore $R'(y,z)=(y,w(y,z))$ is a solution of the YBE for every group $G$.
	
	Suppose further that  $w(x,y,z)$ depends on $x$. We will consider 3 cases:
	
	{\it Case 1}: $w(x,y,z)=w(x,z)$.\\
	In this case the equality (\ref{eqWord}) is rewriting as
	\begin{align}
	w(w(x,z),w(y,q))=w(y,w(z,q)), \label{case2a}
	\end{align}
and  it is easy to see that the left side of this equality depends on $x$. So this case is not possible.
	
	{\it Case 2}: $w(x,y,z)=w(x,y,z)$, i.e. $w$ depends on all generators.\\
In this case we can assume that  
$$
w(x,y,z) \equiv z^{\a_1}g_1(x,y)z^{\a_2}\ldots g_n(x,y)z^{\a_{n+1}}, ~~\alpha_j\in \mathbb{Z},~\mbox{and only}~\a_1\,  \mbox{and/or}~ \a_{n+1}~\mbox{can be zero}.
$$
	So we can write the left side of (\ref{eqWord}):
	$$(w(y,p,q))^{\a_1}g_1(w(x,y,z),w(x,t,p))(w(y,p,q))^{\a_2}\ldots g_n(w(x,y,z),w(x,t,p))(w(y,p,q))^{\a_{n+1}}.$$  
	Note that for every $g_i(x_1,x_2)$ the word $g_i(w(x,y,z),w(x,t,p))$ depends on $t$ and $z$. Moreover since $(w(y,p,q))^{\a_i}$ depends on $q$ for every nonzero $\a_i$, we obtain that the words $g_i(w(x,y,z),w(x,t,p))$ must not depend on $x$. It follows from Lemma \ref{dependLemma2} that if the word $g_i(w(x,y,z),w(x,t,p))$ does not depend on $x$, then $w=x^{\g}h(y,z)$ or $w=h(y,z)x^{\g}$ for some $0\ne \g\in \mathbb{Z}$ and $h(x,y)\in F_2(x,y)$. 
	
	Suppose that $w=x^{\g}h(y,z)$. If there exist $g_{i_0}=y^k$  for some $k\in \mathbb{Z}$, then $g_{i_0}(w(x,y,z),w(x,t,p))$ depends on $x$ by Lemma \ref{reducedLemma}. We get the contradiction. Therefore $h(y,z)=y^{\a}z^{\b}$, where $\a,\b\in\mathbb{Z}\backslash\{0\}$. So in this case we obtain
	$$(x^{\g}y^{\a}z^{\b})^{\g}(x^{\g}t^{\a}p^{\b})^{\a}(y^{\g}p^{\a}q^{\b})^{\b}=y^{\g}t^{\a}(z^{\g}p^{\a}q^{\b})^{\b}.$$
	Cancellation of $x$ can occur only if $\g=-1,\a=1$, then
	$$z^{-\b}y^{-1}tp^{\b}(y^{-1}pq^{\b})^{\b}=y^{-1}t(z^{-1}pq^{\b})^{\b}.$$
	Since $\b\ne0$ we get the contradiction. 
	
	By analogy we can show that the case $w=h(y,z)x^{\g}$ also is not possible.
	
	{\it Case 3}:  $w(x,y,z)=w(x,y)$.\\
	In this case equality (\ref{eqWord}) is rewriting as
	\begin{align}
	w(w(x,y),w(x,t))=w(y,t).\label{case2c}
	\end{align}
	We know that $l(w(x,y)) = k\ge 2$, and at first we assume that $l_x(w(x,y)) = m\ge 2$. Then 
$$
l_x(w(w(x,y),w(x,t)))\ge 2(m-1) + (m-2)(k-2)\ge 2(m-1) > 0.
$$
Therefore  $w(w(x,y),w(x,t))$ depends on $x$. So we obtain that $l_x(w(x,y)) = m = 1$ or equivalently $w(x,y)\equiv y^ax^by^c$, where $b\ne0$.
	
	It is not difficult to prove that $w(x,y) \equiv y^ax^by^c$ satisfies (\ref{case2c}) if and only if $w \equiv x^{-1}y$ or $w \equiv yx^{-1}$.
\end{proof}

\begin{remark}
	If $P_{13}(x, y, z) = (z, y, x)$ is a permutation of the first and the third components and  $R(x,y,z) = (x, y, w(x,y,z))$ is an elementary 3-solution to the TE, then
	$$
	P_{13} R P_{13} (x, y, z) = (w(z,y,x),y,z)
	$$
	is an elementary 1-solution to the TE. 
\end{remark}

A description of verbal elementary 2-solutions gives

\begin{theorem}
	Let $R\colon G^3\to G^3$ be a verbal elementary 2-solution of the TE for every group $G$, then one of the following cases holds:
	\begin{enumerate}
		\item $R(x,y,z)=(x,w(x,z),z)$, where $w \equiv xz$ or $w \equiv zx$,
		\item $R(x,y,z)=(x,w(x,y),z)$, where $R'(x,y)=(x,w(x,y))$ is a solution of the YBE for every group $G$,
		\item $R(x,y,z)=(x,w(y,z),z)$, where $R'(y,z)=(y,w(z,y))$ is a solution of the YBE for every group $G$.
	\end{enumerate}
\end{theorem}

\begin{proof}
	If 
$$
R(x,y,z) =  (x,w(x,y,z),z),~~~x, y, z \in G,
$$
 is a solution of the TE, then by Corollary \ref{c2}
	\begin{align}
	\label{eqWord*}
	w(w(x,y,z),w(x,t,p),q)=w(x,w(y,t,q),w(z,p,q)).
	\end{align}
	
	Suppose at first, that $w$ dos not depend on $x$, i.e.  $w(x,y,z)=w(y,z)$.	In this case the equality (\ref{eqWord*}) takes the form
	$$w(w(t,p),q)=w(w(t,q),w(p,q)).$$ 
	Therefore $R'(y,z)=(y,w(z,y))$ is a solution of the YBE for every group $G$.
	
	Suppose further that  $w(x,y,z)$ depends on $x$. We will consider 3 cases:
	
	{\it Case 1}: $w(x,y,z)=w(x,z)$.\\
	In this case equality (\ref{eqWord*}) is rewriting as
	\begin{align}
	w(w(x,z),q)=w(x,w(z,q)).\label{case2a*}
	\end{align}
	Let $l_2(w(x_1,x_2))=k\ge1$ and $l_1(w(x_1,x_2))=m\ge1$. From equality (\ref{case2a*}) and Lemma \ref{reducedLemma} we have $$m=l_x(w(x,w(z,q)))=l_x(w(w(x,z),q))\ge(k-1)m.$$ Therefore $l_2(w(x_1,x_2))\in \{1,2\}$. Similarly we get $l_1(w(x_1,x_2))\in \{1,2\}$.
	So, using direct calculations in this case we obtain two solutions: $w(x,z) \equiv xz$ and $w(x,z)  \equiv zx$.
	
	{\it Case 2}: $w(x,y,z)=w(x,y,z)$.\\
	Denote by $m=l_1(w(x_1,x_2,x_3))$ and $n=l_3(w(x_1,x_2,x_3))$. Then 
	$$
	w \equiv \omega_1(x_1,x_2) \, x_3^{\a_1} \, \omega_2(x_1,x_2)\ldots x_3^{\a_n} \, \omega_{n+1}(x_1,x_2).
	$$
	
	Let us assume that $m\ge3$. It is easy to see that $$m=l_x(w(x,w(y,t,q),w(z,p,q)))=\sum_{i=1}^{n+1}l_x(\omega_i(w(x,y,z),w(x,t,r))).$$
	Let us assume that there is a subword $\omega_j(x_1,x_2)$ of $w$ such that $k=l(\omega_j(x_1,x_2))>1$. Then by Lemma \ref{dependLemma} we have $$m\ge l_x(\omega_j(w(x,y,z),w(x,t,r)))\ge 2(m-1)+(m-2)(k-2)\ge 2m-2> m.$$
	So, all $\omega_i(x_1,x_2)$ have syllable length one i.e. $\omega_i(x_1,x_2)  \equiv x_1^{\a_i}$ or $\omega_i(x_1,x_2)  \equiv x_2^{\b_i}$. By Lemma \ref{reducedLemma} we have $l_x(\omega_j(w(x,y,z),w(x,t,r)))\ge m$. Therefore $l+1=1$ in sum above, but the word $w$ depends on $x_3$. Hence $m\in\{1,2\}$.
	
	Similarly, it is established that $n\in\{1,2\}$. We get a finite number of possible variants of the word $w$. By checking it is easy to make sure that none of them is a solution.
	
	{\it Case 3}:  $w(x,y,z)=w(x,y)$.\\
	In this case the equality (\ref{eqWord}) is rewriting as
	\begin{align}
	w(w(x,y),w(x,t))=w(x,w(y,t)).\label{case2c*}
	\end{align}
	Therefore $R'(x,y)=(x,w(x,y))$ is a solution of the YBE for every group $G$.
\end{proof}

From these results follows that it is not possible to find  non-trivial invertible verbal elementary  solutions to the TE on free non-abelian groups. We do not know are there exist non-elementary solutions of this type.
On the other side, I.~G.~Korepanov constructed non-trivial invertible solutions to the $4$-SE on free non-abelian (and, hence, on arbitrary) group.

\begin{example}[I. G. Korepanov] \label{Kor}
Let $G$ be an arbitrary group and $a, b$ are two fixed elements of $G$. Then the map
$$
(x, y, z, w) \mapsto (y w^{-1} a, x b z, w, z),~~x, y, z, w \in G,
$$ 
defines a solution to the $4$-SE on $G$. It is easy to see that if $a=b=1$, then this solution is a verbal one.
\end{example}

We studied verbal universal elementary solution for the TE. Also, it is possible to study solutions in some classes of groups, abelian, nilpotent, solvable and so on.

\begin{definition}
Let $R : F_n \to F_n$ be an endomorphism of the free group that is defined on the generators,
$$ 
R(x_1, x_2, \ldots, x_n) = (w_1(x_1, x_2, \ldots, x_n),  w_2(x_1, x_2, \ldots, x_n), \ldots, w_n(x_1, x_2, \ldots, x_n)),
$$
$\mathcal{V}$ be  a variety of groups. We will say that $R$ defines a {\it verbal universal} $\mathcal{V}$-solution for {\SE}, if $(G, R)$ is a solution 
for any group $G \in \mathcal{V}$. If $\mathcal{V}$ is variety of abelian groups, we say on universal verbal abelian solution, if $\mathcal{V}$ is variety of nilpotent groups, we say on universal verbal nilpotent solution, and so on.
\end{definition}

\subsection{Verbal solutions in abelian groups} To find verbal universal abelian solutions for {\SE} one can consider free abelian group of rank $n$.
Verbal solutions of the {\SE} for the free abelian group $\Z^n$ are in one-to-one correspondence with linear solutions over the ring $\Z$. Any linear solution $R$ over $\Z $ may be written by matrix:
	$$R(x) = M x, \quad x \in \Z^n, \ M \in \mathcal{M}_n(\Z).$$

\begin{proposition}\label{transposition} If a map $R$ given by a matrix $M$ is a solution of the {\SE}, then a map, defined by a transposed matrix $M^T$ is also a solution of the {\SE}.
\end{proposition}
Fortunately in some cases we already know solutions of  {\SE}. For example, all linear bijective solutions for the {\SE[2]},  {\SE[3]} and  {\SE[4]}
are listed in \cite{HLinear}.
Thus we  can list all verbal universal abelian  solutions in these cases. 

\begin{proposition}
	Let $G$ be an abelian group, $\a, \b \in \mathbb{Z}$.
	Then all verbal solutions to the YBE are given by following mappings:
	$$
	\begin{aligned}
		(a, b) & \mapsto (a^{\a}, b^{\b});
		\\
		(a, b) & \mapsto (b, a);
		\\
		(a, b) & \mapsto (a^{\a}b^{1 - \a \b}, b^{\b});
		\\
		(a, b) & \mapsto (a^{\a}, a^{1 - \a \b} b^{\b}),
	\end{aligned}
	$$
	where $a, b \in G$.
\end{proposition}
\begin{proposition}
	Let $G$ be an abelian group, $\a, \b, \g \in \mathbb{Z}$.
	Then all verbal solutions to the TE come either from $ k $-amalgams of solutions of the YBE and 1-simplex equations or 
	from the list of mappings below by conjugation from  Proposition~\ref{invsymm} or transposition from  Proposition~\ref{transposition}.
	$$
	\begin{aligned}
		(a, b, c) & \mapsto (a^{\a}, a, a^{-\b} b c^{\b});
		\\
		(a, b, c) & \mapsto (b, a, a^{-\a} b c^{\a});
		\\
		(a, b, c) & \mapsto (a^{\a} b^{1 - \a \b} c^{\a (\b \g - 1)}, b^{\b} c^{1 - \b \g}, c^{\g});
		\\
		(a, b, c) & \mapsto (a^{\a} b^{1 - \a \b} c^{\g (\a \b - 1)}, b^{\b} c^{1 - \b \g}, c^{\g});
	\end{aligned}
	$$
	where $a, b, c \in G$.
\end{proposition}

\begin{proposition}
	Let $G$ be an abelian group, $\a, \b, \g, \d \in \mathbb{Z}$.
	Then all verbal solutions to the 4-SE are obtained either from $ k $-amalgams of solutions of m-SE where $m < 4$  or from the mappings below by conjugation from  Proposition~\ref{invsymm} and matrix transposition from  Proposition{~\ref{transposition}}.
	$$
	\begin{aligned}
		(a, b, c, d) & \mapsto (b d^{-1}, a c, d, c);
		\\
		(a, b, c, d) & \mapsto (b c^{-\a} d^{\a\b-1}, a c, c^{\a} d^{1 - \a\b}, d^{\b});
		\\
		(a, b, c, d) & \mapsto (b, a, a^{-\a} b c^{\a}, a^{\a\b-1}c^{1-\a\b}d^\b);
		\\
		(a, b, c, d) & \mapsto (b c^{-\a} d^{\a\b}, a c d^{-\b}, c^{\a} d^{1-\a\b}, d^{\b});
		\\
		(a, b, c, d) & \mapsto (b, a, a^{-\a}bc^{\a},a^{\a\b}b^{-\b} c^{1-\a\b} d^{\b});
		\\
		(a, b, c, d) & \mapsto (a^{\a} b c^{-\a}, c, b, b^{-\b} c d^{\b});
		\\
		(a, b, c, d) & \mapsto (a^{\a}, b c d^{-\b}, a^{-\a} b d, d^{\b});
		\\
		(a, b, c, d) & \mapsto (a^{\a} b c^{-\a} d^{\a\b}, c d^{-\b}, b d, d^{\b});
		\\
		(a, b, c, d) & \mapsto (a^{\a}b^{1-\a\b}c^{\g(\a\b-1)}d^{\g\d(1-\a\b)}, b^{\b}c^{1-\b\g}d^{\d(\b\g-1)}, c^{\g}d^{1-\g\d}, d^{\d});
		\\
		(a, b, c, d) & \mapsto (a^{\a}b^{1-\a\b}c^{\a(\b\g-1)}d^{\a\d(1-\b\g)}, b^{\b}c^{1-\b\g}d^{\d(\b\g-1)}, c^{\g}d^{1-\g\d}, d^{\d});
		\\
		(a, b, c, d) & \mapsto (a^{\a}b^{1-\a\b}c^{\a(\b\g-1)}d^{\a\b(1-\g\d)}, b^{\b}c^{1-\b\g}d^{\b(\g\d-1)}, c^{\g}d^{1-\g\d}, d^{\d}),
		\\
	\end{aligned}
	$$
	where $a, b, c, d \in G$.
\end{proposition}

\bigskip

\section{Research problems}

At the end, we would like to suggest some problems for further research.

\begin{problem}
\begin{enumerate}\ 
\item Find verbal universal solutions for $n$-simplex equations in class of all groups.

\medskip

\item We know that braid group corresponds to the YBE. What group corresponds to $n$-simplex equations?

\medskip

\item With the  YBE are connected braided categories. What categories correspond to $n$-simplex equations?

\medskip

\item In \cite{BNas, BNas1} were defined multi-switches for using  virtual knot invariants. Is it possible to use them for construction  solutions for  parametric YBE?

\medskip

\end{enumerate}
\end{problem}

\section*{Acknowledgements} 
The results of the present paper were obtained during the program on $n$-simplex equations  advised by V.~Bardakov and T.~Kozlovskaya in the frame of The Second Workshop at the Mathematical Center in Akademgorodok supported by the Tomsk Regional Mathematical Center. The authors thank Dmitriy Talalaev and Mahender Singh for interesting lectures and fruitful discussions. Also, we thank I.~G.~Korepanov for permission to include his example (Example \ref{Kor}) in our text.

\bigskip

Valeriy Bardakov$^{\dag, \ddag, **}$ (bardakov@math.nsc.ru),\\

Bogdan Chuzinov$^{*,\dag}$ (b.chuzhinov@g.nsu.ru),\\

Ivan  Emel'yanenkov$^*$ (i.emelianenkov@g.nsu.ru),\\

Maxim Ivanov$^*$ (m.ivanov2@g.nsu.ru),\\

Tatyana Kozlovskaya$^{ \ddag}$ (t.kozlovskaya@math.tsu.ru),\\

Vadim Leshkov$^{*}$ (v.leshkov@g.nsu.ru)

~\\
$^*$ Novosibirsk State University, Pirogova 1, 630090 Novosibirsk, Russia,\\
$^{\dag}$ Tomsk State University, pr. Lenina 36, 634050 Tomsk, Russia,\\
$^{\ddag}$ Sobolev Institute of Mathematics, Acad. Koptyug avenue 4, 630090 Novosibirsk, Russia,\\
$^{**}$ Novosibirsk State Agrarian University, Dobrolyubova  160, 630039 Novosibirsk,  Russia.\\

\end{document}